\documentclass[aps,prb,
twocolumn,longbibliography,notitlepage,superscriptaddress,nofootinbib]{revtex4-1}

\bibpunct{[}{]}{;}{n}{}{}
\usepackage[caption=false]{subfig}

\usepackage{myStyle}
\def\eps{\epsilon}

\usepackage{ gensymb }
\usepackage{xfrac}
\usepackage{slashed}
\usepackage{framed}
\usepackage{hhline}
\usepackage{footnotebackref}

\usepackage{tcolorbox}
\usepackage{cleveref}
\usepackage{graphicx}
\usepackage{stmaryrd}
\usepackage[mathscr]{euscript}

\usepackage{subfiles}

\def\StateList{\{\text{Sym},\text{SP},\text{VP},\text{SVP}\}}
\newcommand\utimes{\mathbin{\ooalign{$\sqcup$\cr \hfil\raise0.42ex\hbox{$\scriptscriptstyle\times$}\hfil\cr}}}

\def\Nel{N_{\text{el}}}

\def\eF{E_{\text{F}}}

\def\diag{\mathrm{diag}}

\DeclareMathOperator\arcsinh{arcsinh}
\DeclareMathOperator\arctanh{arctanh}
\DeclareMathOperator\Kell{K}
\DeclareMathOperator\Eell{E}

\def\kF{k_{\text{F}}}

\def\Ksc{\mathscr{K}}

\def\RPA{\text{RPA}}

\def\kk{\text{k}}
\def\rr{\text{r}}
\def\vv{\text{v}}
\def\ss{\text{s}}
\def\xx{\text{x}}
\def\cc{\text{c}}
\def\xc{\text{xc}}
\def\xr{\text{xr}}

\def\nel{n_{\rm el}}
\def\ee{\mathrm{e}}
\def\ii{\mathrm{i}}
\def\rs{r_\mathrm{s}}
\def\Area{\text{A}}

\def\aB{a_{\mathrm{B}}}

\def\Dsc{\mathscr{D}}
\def\Esc{\mathscr{E}}
\def\Fsc{\mathscr{F}}

\def\Xsc{\mathscr{X}}
\def\Vsc{\mathscr{V}}
\def\Qsc{\mathscr{Q}}

\begin{document}

\title{
Theory of Coulomb driven nematicity in a multi-valley two-dimensional electron gas
}

\date{\today}
\author{Vladimir Calvera}
\affiliation{Department of Physics, Stanford University, Stanford, CA 94305, USA}
\author{Agnes Valenti}
\affiliation{Institute for Theoretical Physics, ETH Zurich, CH-8093, Switzerland}
\affiliation{Center for Computational Quantum Physics, Flatiron Institute, New York, NY, 10010, USA}
\author{Sebastian D. Huber}
\affiliation{Institute for Theoretical Physics, ETH Zurich, CH-8093, Switzerland}

\author{Erez Berg}
\affiliation{Department of Condensed Matter Physics, Weizmann Institute of Science, Rehovot 76001, Israel}

\author{Steven A. Kivelson}
\affiliation{Department of Physics, Stanford University, Stanford, CA 94305, USA}
\begin{abstract}
 The properties of a two-dimensional electron gas (2DEG) in a semiconductor host with two valleys related by an underlying $C_4$ rotational symmetry are studied using Hartree-Fock (HF) and various other many-body approaches. A familiar artifact of the HF approach is a degeneracy between the valley polarized - ``Ising nematic'' - and spin polarized - ferromagnetic - phases, which is inconsistent with recent variational Monte Carlo (VMC) results. Correlation effects, computed either within the random phase approximation (RPA) or the T-matrix approximation, enhance the valley susceptibility relative to the spin susceptibility. Extrapolating the results to finite interaction strength, we find a direct first-order transition from a symmetry-unbroken state to a spin unpolarized Ising nematic fluid with full valley polarization, in qualitative agreement with VMC. The RPA results are also reminiscent of experiments on the corresponding 2DEG in AlAs heterostructures. 
\end{abstract}
\maketitle

\section{Introduction}

The two-dimensional electron gas (2DEG) is one of the cornerstone problems in condensed matter physics. Its theoretical importance stems from the  simplicity of the model, the richness of the phase diagram, and the fact that its properties can be accurately probed experimentally. 
Still further diversity of behavior becomes possible when, in addition to the spin degree of freedom, the electrons carry a ``valley'' degree of freedom that depends on the host semiconductor's band structure. 

In particular, we consider a system that is overall $C_4$ symmetric but each valley is not $C_4$ symmetric on its own 
(Fig.~\ref{fig:Fig1}a).
Such a 2DEG is realized in AlAs quantum wells~\cite{Shayegan2006AlAs2DEG}. In that system, there is experimental evidence \cite{Hossain2021AlAsValleyPolarization} that in the low temperature ($T$) limit, this generalized 2DEG has a direct  transition from a 
{symmetry-unbroken} Fermi liquid
to a 
nematic, fully valley-polarized Fermi liquid. 
Moreover, 
we~\cite{valenti2023nematic} recently 
used variational Monte Carlo (VMC) to numerically study possible nematic and ferromagnetic fluid phases of the two-valley 2DEG 
and obtained results that are qualitatively consistent with the experiments, {i.e.} we found a {direct} transition between a $C_4$ symmetric liquid to a fully valley-polarized nematic liquid state without any spin ordering. The present work aims to provide an analytical understanding of 
these results. 

The standard Hartree-Fock (HF) approximation is inadequate for this purpose, both for the well known reason that it implies a Stoner-like instability at much too large electron densities, and because it implies equal tendencies to spin and valley (nematic) polarized states. We show that this unphysical degeneracy is lifted by the correlation energy at second order in interactions, favoring valley polarization. 
{ For short-range interactions, this result holds in T-matrix approximation, but for the case of Coulomb interactions, which is our principal focus, it is necessary to carry through an infinite-order resummation of perturbation theory within the Random Phase Approximation (RPA).}  We find that the RPA captures the qualitative aspects of the anisotropy-density phase diagram obtained with VMC (see Fig.~\ref{fig:Fig1}b). We also considered the possibility of
an inter-valley coherent state (IVC), but tentatively interpret the fact that the non-interacting IVC susceptibility is smaller than the spin or valley ones as evidence that it can be neglected here.

\begin{figure}[h!]
    \centering
    \includegraphics[width= 0.37\textwidth]{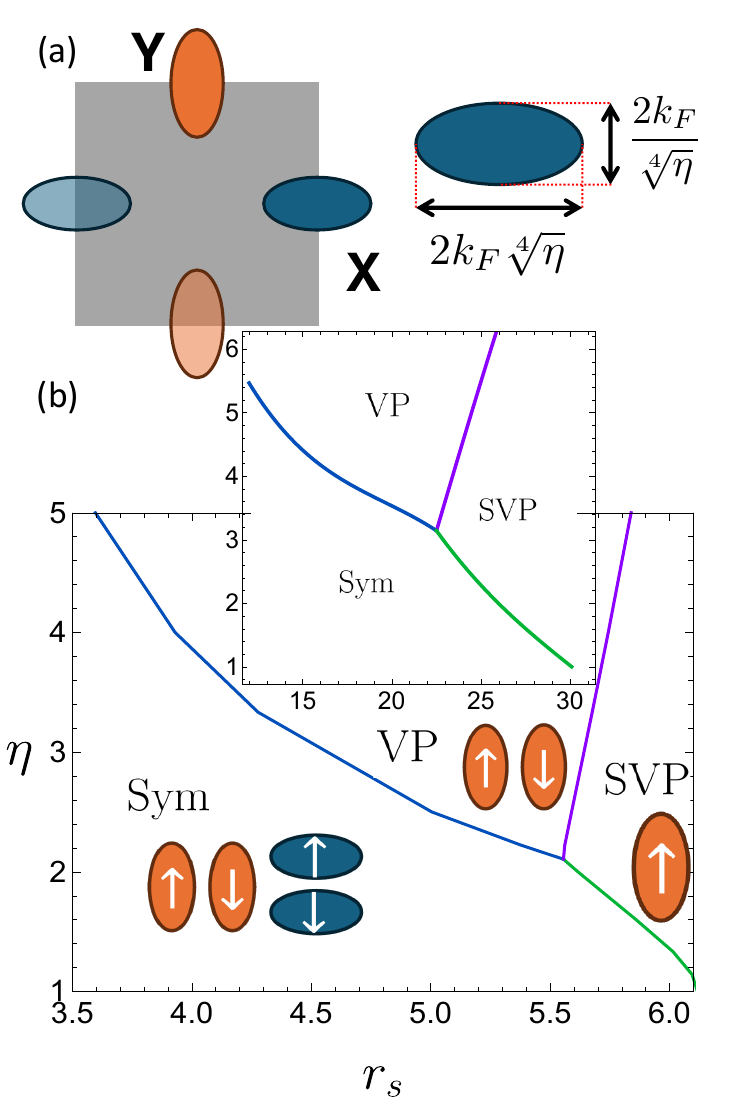}
    \caption{\textbf{(a)}
    Sketch of the position of the band minima of a two-dimensional semiconductor together with the definition of mass anisotropy $\eta$, and the Wigner-Seitz radius $\rs$. \textbf{(b)} Random Phase Approximation phase diagram. Inset: Variational Monte Carlo phase diagram from \cite{valenti2023nematic}.
    }
    \label{fig:Fig1}
\end{figure}

This paper is organized as follows: In Sec.~\ref{sec:Setup} we present the details of the model; Sec.~\ref{sec:Hartree-Fock} presents the results of the Hartree-Fock analysis; In Sec.~\ref{sec:LiftDegeneracy} both the effects of higher order perturbation theory and the RPA results are presented; Sec.~\ref{sec:DiluteLimitShortRange} deals with the dilute limit with short-range interactions; and we end in Sec.~\ref{sec:Discussion} with a discussion of our results.

\section{Setup}\label{sec:Setup}

The model we consider consists of interacting electrons in the presence of a non-dynamical uniform positively charged background. In addition to the spin index, each electron carries a valley index, $\tau=\pm 1$. 
The Hamiltonian is given by:
\begin{align}
\label{eq:Hamiltonian}
H=&-\sum \limits_{i} \frac{1}{2m}\big(\eta^{-\tau_i/2} \partial_{i,x}^2+ \eta^{+\tau_i/2} \partial_{i,y}^2 \big) \nonumber \\ &+ \sum \limits_{i<j} V(|{\bf r}_i- {\bf r}_j|),
\end{align}
where $i,j$ run over all electrons, $\tau_i$  denotes the valley flavor of electron $i$, and the (geometric) mean band mass is $m$. $V(r)= \frac{e^2}{4\pi \epsilon r}$ is the Coulomb interaction, where $\epsilon$ is the dielectric constant of the medium. (In Sec. \ref{sec:DiluteLimitShortRange} we will consider the case of a short-ranged effective interaction.) The parameter $\eta$ controls the anisotropy of the system: $\eta\neq 1$ corresponds to an elliptical shape of the Fermi surface. However, the model always has a $C_4$ rotation symmetry that acts on the spatial coordinates and exchanges the valley flavor (See Fig. \ref{fig:Fig1}a.)

Note that in Eq.~\eqref{eq:Hamiltonian},
consistent with the effective mass approximation, we have included only the small momentum-transfer portions of the interaction, with the consequence that $V$ (but not $H$) is $\SU(4)$ symmetric in spin and valley space. In an actual semiconductor, there are also terms that involve momentum transfers comparable to the Brillouin zone size, such as inter-valley exchange and pair hopping interactions; these are typically small relative to the $\SU(4)$ symmetric part of the Coulomb interaction (which involves a momentum transfer of the order of the Fermi momentum) by the ratio of the lattice spacing to the inter-electron distance, which is of the order of $10^{-2}$ in systems such as AlAs \cite{valenti2023nematic}.

There are only two dimensionless parameters: the mass anisotropy $\eta$ and the dimensionless Wigner-Seitz radius $\rs = \tfrac{m e^2}{4\pi \epsilon \hbar^2}\tfrac{1}{\sqrt{\pi \nel}}$, where $\nel$ is the {total} electron density. For ease of exposition, we introduce a flavor index $\alpha=1,2,3,4$ that corresponds to $(\sigma,\tau) = (+1,+1), (-1,+1), (+1,-1), (-1,-1)$, respectively, where $\sigma$ denotes the spin. We also define the anisotropy of flavor $\alpha$, $\eta_\alpha$, {as the effective mass anisotropy of that flavor,} so that $\eta_1=\eta_2 = \eta$ and $\eta_3=\eta_4 = 1/\eta$. 

\section{Hartree-Fock analysis}\label{sec:Hartree-Fock}

We have carried out a Hartree-Fock analysis of the Hamiltonian in Eq. (\ref{eq:Hamiltonian}), restricting ourselves to spatially uniform states, in which the number of electrons of each flavor is conserved. 
Our variational parameters are thus the occupancy and degree of spin polarization (along an arbitrarily chosen axis) of each valley, and the {variational anisotropy}, $\bar{\eta}_\alpha$, of each occupied elliptical Fermi surface {{that is not necessarily equal to the bare anisotropy}}.

The resulting HF phase diagram is shown in Fig.~\ref{fig:HF-phase diagram}. 
We find four distinct phases in each of which the HF ground state has an integer number $1\leq n \leq 4$ of equally occupied flavors - corresponding to the phases labeled $n$F in the figure. Note that the 2F region encompasses a variety of possible degenerate phases with distinct patterns of symmetry breaking, of which a subset (shown schematically) is:
1) a valley-polarized (VP) state (e.g. $\nu_1=\nu_2=1/2$); 2) a $C_4$-symmetric spin-polarized (SP) state (e.g. $\nu_1=\nu_3=1/2$); and 3) an altermagnet (AM)~\cite{altermagnetism} (e.g. $\nu_1=\nu_4=1/2$). 
Here $\nu_\alpha$ is the fraction of electrons with flavor $\alpha$, such that $\sum_\alpha\nu_\alpha=1$. 
{Within HF, all the above 2F phases are exactly degenerate, despite the fact that these states are not related by any symmetry of the problem. Correlation effects beyond HF lift this degeneracy, as we show below.} In all cases, we find that the optimal value of $\tilde{\eta}_{\alpha}$ is closer to 1 than the bare anisotropy, since the exchange contribution to the energy favors a less eccentric Fermi surface. The details of the calculations are given in App.~\ref{app:Hartree-Fock analysis}. 

\begin{figure}[b]
    \centering
    \includegraphics[width=0.45\textwidth]{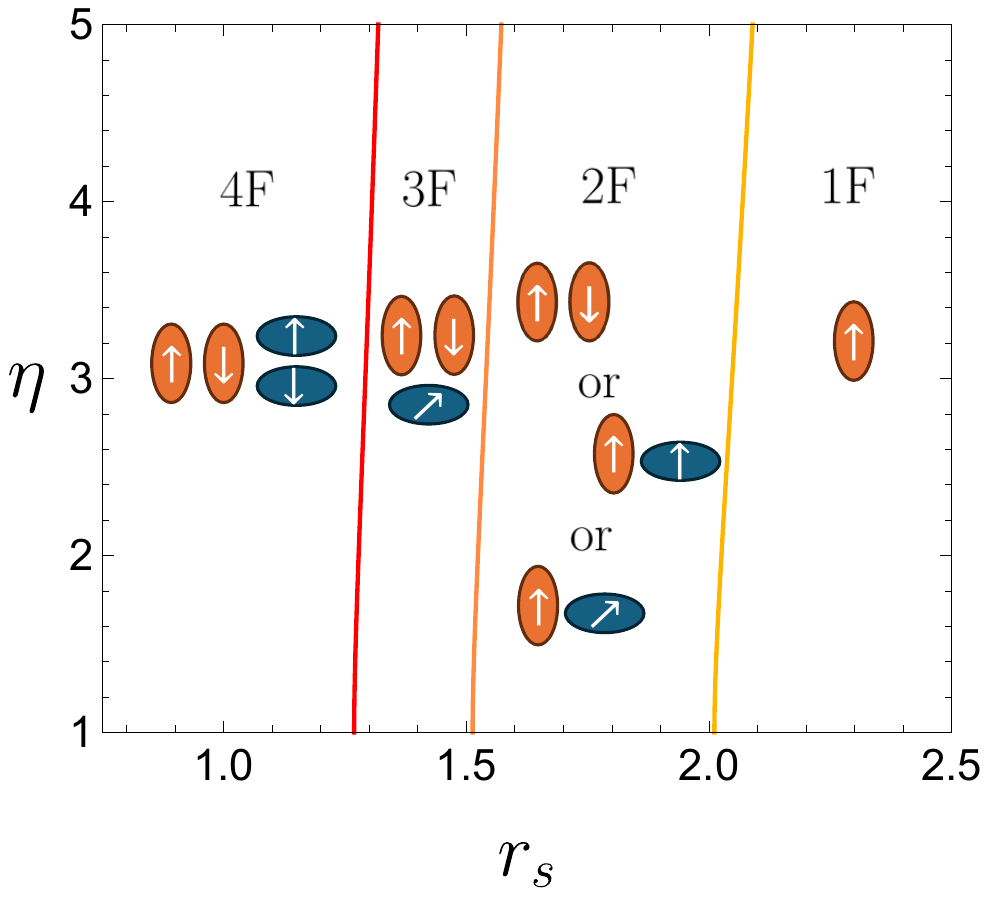}
    \caption{Hartee-Fock phase diagram.  `$n$F' denote phases where $n=1,2,3$, or $4$ 
    flavors are equally populated. {The set of degenerate 2F states {shown are} the valley-polarized state, spin-polarized state, and a state in which the relative spin orientation between valleys is arbitrary. }
 }
    \label{fig:HF-phase diagram}
\end{figure}
{We have also analyzed, albeit incompletely,
inter-valley coherent (IVC) states that preserve $C_4$ rotation symmetry (see App.~\ref{app:IVCstates}). We found them to be disfavored compared to the states above. In particular, we find that the IVC susceptibility for the symmetric state is always smaller (for $\eta\neq 1$) than $\chi_{\rm{v}}$ 
\begin{equation}
   \frac{\chi_{\rm{IVC}}}{\chi_{\rm{v}}} =   \frac{2\sqrt{\eta}}{\eta-1}\arcsin\left(\frac {\eta-1}{\eta +1}\right) \leq 1,
\end{equation}
where $\chi_{\rm{v}} = \frac{2m}{\pi}$ is the valley susceptibility. For intermediate interaction strength, we have estimated the energy of the IVC and SVP states by extrapolating the expressions obtained from a strong-coupling expansion for Slater determinants describing a $C_4$-symmetric IVC state and the SVP state, respectively. We found that the SVP state has lower energy than the IVC state.} 

A more systematic self-consistent Hartree-Fock calculation, including the possibility of states that break the translational symmetry, is left for future work.

\section{Lifting the degeneracy}\label{sec:LiftDegeneracy}

\subsection{Second-order perturbation theory}
\begin{figure}[t]
    \centering
    \includegraphics[width=0.45\textwidth]{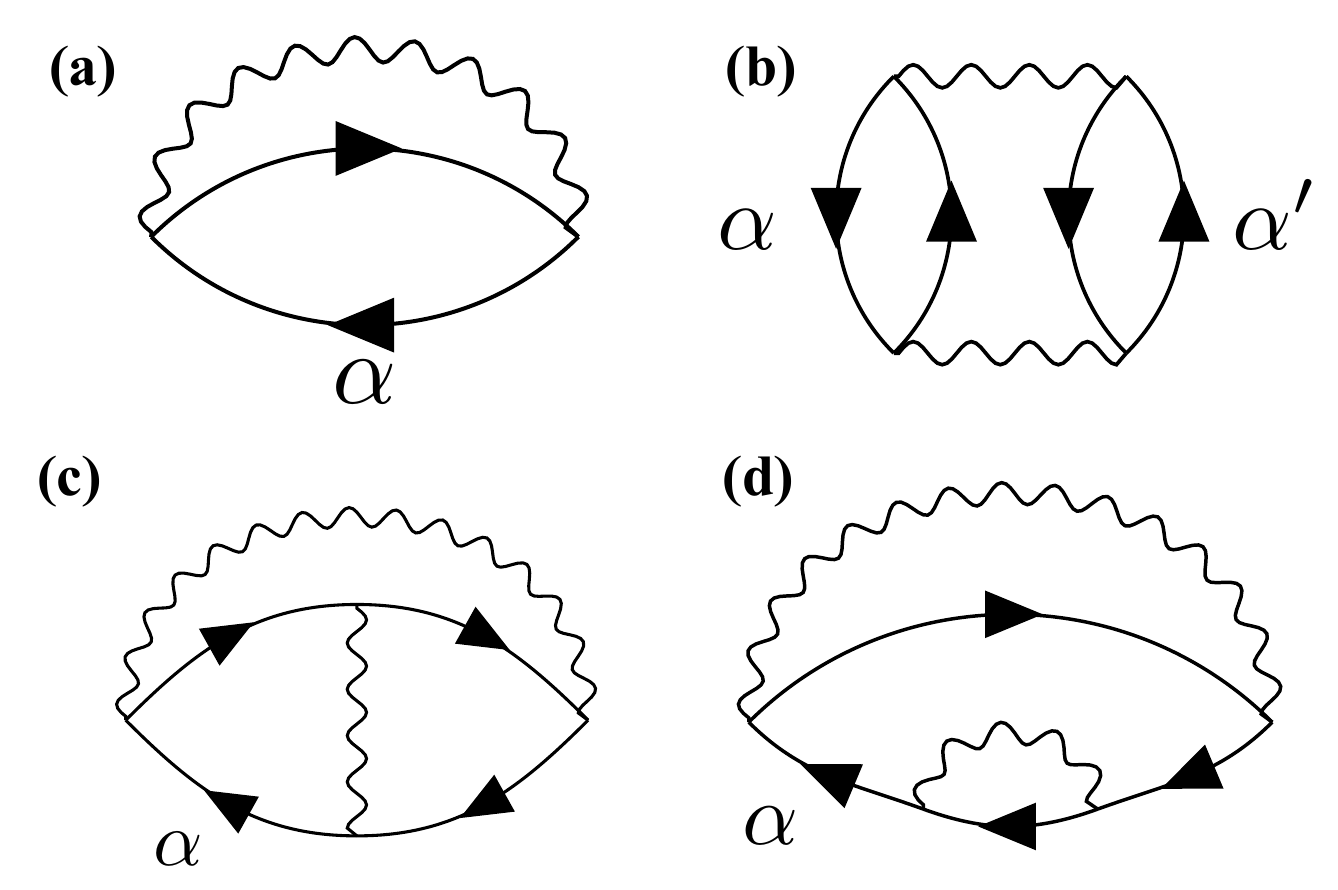}
    \caption{Diagrams that contribute the energy to second order in interactions. $\a,\a'\in \{(X,\ua), (X,\da), (Y,\ua), (Y,\da)\}$ are flavor indices.
    }
    \label{fig:2PT}
\end{figure}
To gain intuition for why the valley polarized state is preferred over the other states, we evaluated corrections beyond HF for the various 2F states. Using the fact that the particle number for each flavor $\alpha$ is a good quantum number, we can evaluate the ground state energy with any given flavor polarization using many-body perturbation techniques. Due to the spin-rotation symmetry within each valley, without loss of generality, we can focus on two states: the valley polarized (VP) and the spin polarized (SP) state. 

The diagrams corresponding to corrections to the ground state energy up to second order are shown in Fig.~\ref{fig:2PT}.
The only diagram that gives a different contribution to the energy of VP and SP states is diagram \textbf{b}.  This is because terms with a single fermion loop are sums of individual flavor terms, and are not sensitive to the different dispersion of the two valleys. This is not the case for diagram \textbf{b} because there are two fermion loops. 
Its contribution to the energy, $E_{\textbf{b}}$ (the subscript refers to the diagram), is 
\begin{align}\label{eq:E_e}
    \frac{E_{\textbf{b}}}{\A} &= -\frac{1}{2}\int_{0}^{\infty}\!\frac{\dd{\Omega}}{2\pi}\int\!\frac{\dd^2{\qbs}}{(2\pi)^2} \left[\tilde{v}(\qbs)\right]^2 [\Pi(\ii\Omega,\qbs)]^2,
\end{align}
where $\A$ is the area of the system and $\Pi(\omega,\qbs)$ {is the Lindhard function evaluated at the corresponding electron densities (this depends on $\nel$ and the polarization pattern, {e.g.} VP or SVP). We give explicit expressions in Eq.~\ref{eq:Lindhard} in App.~\ref{app:Energy2DEGAnisotropic}}. $\tilde{v}(\qbs) = \int V(\abs{\rbs})\ee^{\ii\rbs\cdot\qbs}\dd{\rbs}$ is the Fourier transform of the interaction. In App.~\ref{app:Er_Leading}, {we show that $E_{\textbf{b}}$ is necessarily more negative for the VP state than for the SP state assuming that the interactions are $C_4$ symmetric and $V(\abs{\mathbf{r}})$ is not constant.}
{Therefore under these assumptions, to second order in the interactions, the ground state energy in the valley-polarized sector is lower than in the spin-polarized sector, regardless of whether the interactions are attractive or repulsive. }

The result can be generalized for any amount of polarization: to second order in the interaction strength a partially valley polarized state with $\zeta = \nu_1 + \nu_2 - \nu_3 - \nu_4$ has a lower energy than a spin-polarized state with a spin polarization $\zeta = \nu_1 - \nu_2 + \nu_3 - \nu_4$ (see Appendix \ref{app:Er_Leading} for the proof).

\subsection{Random Phase Approximation:} 
Infrared divergences appear in higher-order corrections to the energy due to the long-range nature of the Coulomb potential. This can be resolved by the clever resummation of the diagrams that goes by the name of the Random phase approximation (RPA). 

\begin{figure*}[t]
    \centering    \includegraphics[width=0.8\textwidth]{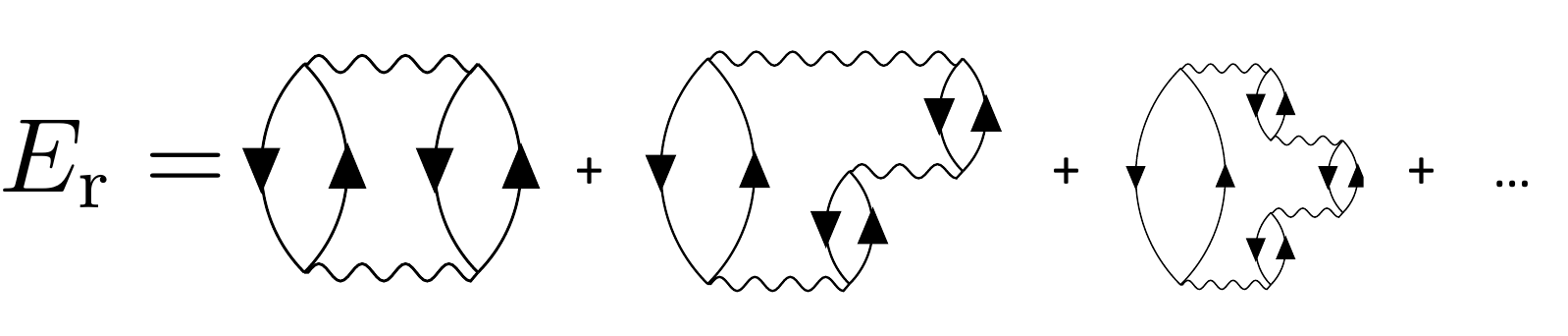}
    \caption{ Diagrammatic representation of the contribution to the correlation energy from ring diagrams, {a.k.a.} the RPA for the correlation energy. }
    \label{fig:EnergyRPA}
\end{figure*}

The correlation energy per unit area within RPA can be written in terms of the non-interacting Lindhard function (summed over all flavors), $\Pi(\ii\Omega,\qbs)$, ({e.g.} see \cite{giuliani2008quantum})
\begin{widetext}
\begin{equation}\label{eq:RPAenergy}
\frac{E_{\rr}}{\A} =
   \int_{0}^{\infty}\frac{\dd{\Omega}}{2\pi}\int\frac{\dd^2{\qbs}}{(2\pi)^2} \Big\{\log[1+\tilde{v}(\qbs)\Pi(\ii\Omega,\qbs)]-\tilde{v}(\qbs)\Pi(\ii\Omega,\qbs)\Big\}.
\end{equation}
\end{widetext}
$E_{\text{r}}$ can be represented diagramatically as in Fig.~\ref{fig:EnergyRPA}. (The subscript `r' references the ring shape of the diagrams.)

Let $E_{\rr}^{\text{pVP}}[\zeta_v]$ and $E_{\rr}^{\text{pSP}}[\zeta_s]$ be the energy in Eq.~\ref{eq:RPAenergy} evaluated using the Lindhard function $\Pi$ of the non-interacting ground state with polarization $\zeta_v=\nu_1 + \nu_2 - \nu_3 - \nu_4$ and $\zeta_s= \nu_1 - \nu_2 + \nu_3 - \nu_4$, respectively. We have the following result:
\begin{theorem}\label{theorem:Valley}
Let $\tilde{v}(\qbs)$ be positive $(\tilde{v}>0)$ and $C_4$ symmetric ($\tilde{v}(C_4\qbs) = \tilde{v}(\qbs)$). The RPA correlation energy for the valley-polarized state and for the spin-polarized state, satisfy \begin{equation}\label{eq:Theorem1Eq}
        {E^{\mathrm{pVP}}_{\rr}}[\zeta] < {E^{\mathrm{pSP}}_{\rr}}[\zeta]
    \end{equation}
    given that $\eta\neq 1$.
\end{theorem}

The intuition for this theorem is that electrons with the same dispersion are more effective at screening each other without paying too much kinetic energy. The proof is given in Appendix~\ref{app:Er_RPA}.

As a corollary, we show that, for Coulomb interactions and within our approximation, the valley susceptibility ($\chi_{\vv}$) is larger than the spin susceptibility ($\chi_{\ss}$). The susceptibilities are given by $\chi_{\text{v/s}}^{-1} \equiv \frac{1}{\A}\pdv[2]{E^{\mathrm{pVP/pSP}}[\zeta]}{\zeta}\eval_{\zeta=0}$. The symmetric state is an extremum of the energy (for sufficiently small $\rs$, it is a local minimum) so $\partial_{\zeta}E^{\mathrm{pSP}}$ and $\partial_{\zeta}E^{\mathrm{pVP}}$ are zero when $\zeta=0$. Thus
\begin{equation}
    \begin{split}
        E^{\text{pSP}} &= E^{\text{Sym}} + \frac{\A \z^2}{2\chi_{\ss}} +\Omc(\z^3);\\
        E^{\text{pVP}} &= E^{\text{Sym}} + \frac{\A\z^2}{2\chi_{\vv}}+\Omc(\z^3);
    \end{split}
\end{equation}
Using Theorem~\ref{theorem:Valley}, we obtain the desired result
\begin{equation}
    0\leq \lim_{\zeta \to 0} \frac{E^{\text{pSP}}-E^{\text{pVP}}}{\A \zeta^2/2} = \frac{1}{\chi_{\ss}}-\frac{1}{\chi_{\vv}} \Longrightarrow \chi_{\ss} \leq \chi_{\vv}.
\end{equation}

Now that we have lifted the {artificial} degeneracy between the valley polarized state and the other 2F states, we examine the phase diagram within RPA. We use the RPA energy $E_{\text{RPA}} = E_{0} + E_{\xx} + E_{\rr}$, where $E_{\text{0}}$ is the non-interacting energy and $E_{\text{x}}$ is the exchange energy (diagram \ref{fig:EnergyRPA}\textbf{a}). {The phase diagram is mapped by minimizing $E_{\text{RPA}}$ over the spin and valley polarizations (allowing for partial polarizations {near the Sym-VP transition line}).}
We show the resulting RPA phase diagram in Fig.~\ref{fig:Fig1}. Satisfactorily, it captures the main features of the variational phase diagram of Ref.~\cite{valenti2023nematic}: 1) only three phases are present: 
{the symmetric phase, the fully valley polarized phase, and the fully polarized spin and valley phases, separated by first-order transitions;} and 2) there is a critical value of anisotropy $\eta$ that stabilizes the valley polarized state. However, the critical value of $\rs$ and $\eta$ are not quantitatively accurate (and, in particular, the critical $r_s$ for valley polarization within RPA is {roughly a factor of 3} smaller than {values from } the numerical VMC calculation), which is to be expected from previous studies in the standard 2DEG \cite{giuliani2008quantum}. 

{We also estimated the valley susceptibility by fitting the polarization dependence of the energy (see App.~\ref{app:PolarizationDependence} for details). We find that, although the susceptibility is significantly enhanced upon increasing $r_s$ (as anticipated recently in Refs.~\cite{raines2024unconventional,raines2024a_unconventional}), it remains finite at the transition. Such an enhancement of the valley susceptibility with decreasing electron density was observed in AlAs~\cite{Shayegan2006ValleyMeasurements}. }

{Finally, we comment on the effects of finite temperature on the phase diagram. We estimated these effects within RPA (see App.~\ref{app:TemperatureEffects}). 
As expected from Fermi liquid theory, the low $T$ of the Helmholtz free energy for each state $a\in \StateList$, $F_{a}$ is $F_{a}[\rs] \approx E_a[\rs] - \gmf_a(\rs) \frac{\A m \p T^2}{12} $. 
Comparing the symmetric and valley polarized phases, $\gmf_{\text{Sym}}(\rs) - \gmf_{\text{VP}}(\rs)>0$, as expected from the fact that the symmetric phase has a higher density of states at the Fermi level. This implies that the Symmetric to VP phase boundary should shift to higher values of $r_s$ with increasing temperature. 
Moreoever, we find that $\gmf_{\text{Sym}}(\rs) - \gmf_{\text{VP}}(\rs)$ is enhanced by interactions relative to the non-interacting value. }

\section{Model with short-range interactions}\label{sec:DiluteLimitShortRange}

\begin{figure*}[t]
    \centering    
    \includegraphics[width=0.8\textwidth]{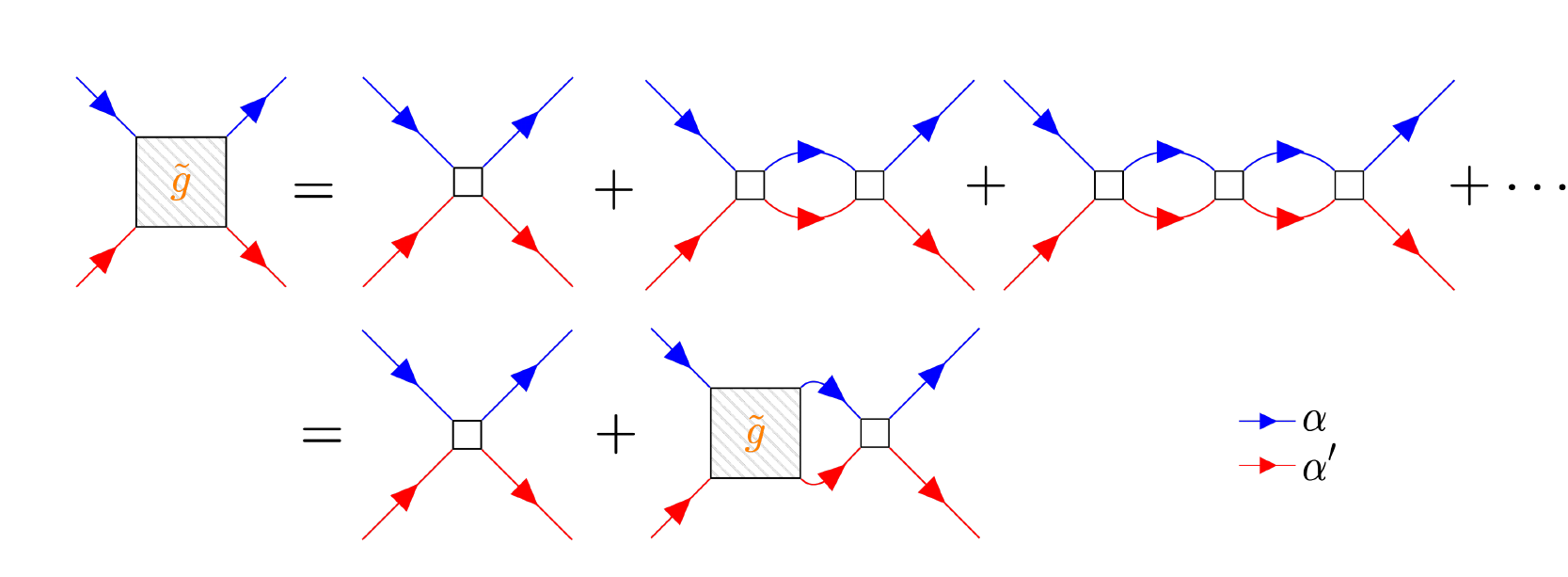}
    \caption{ Diagrammatic representation of T-matrix approach. White boxes represent bare interactions while big boxes denote the renormalized interactions. Different color lines represent different flavors. }
    \label{fig:Tmatrix}
\end{figure*}
In this section, we analyze the behaviour of the same two-valley 2DEG, Eq.~\eqref{eq:Hamiltonian}, where the Coulomb interaction is replaced by a finite-ranged interaction. We consider the dilute limit where the average inter-electron distance is larger than the range of the interaction. This situation can be realized in the presence of nearby metallic gates which screen the long-ranged part of the interaction, if the distance to the gates is shorter than the inter-electron distance. Although this is not the case in the experiments of Ref.~\cite{Hossain2021AlAsValleyPolarization}, this limit is instructive, since as we shall show below, the selection of the valley polarized state by correlation effects driven by the anisotropic dispersion is particularly transparent. The same analysis can be readily generalized to other multi-valley 2DEGs with anisotropic dispersions. 

Specifically, we consider a contact interaction, {i.e.} 
$V({\rbs}) =  \frac{4\pi g_0}{m} \delta(\rbs)$, 
where $g_0>0$ is 
dimensionless. We assume that the overall electronic bandwidth, $W$, is much larger than the Fermi energy, $E_F$. 
We use the T-matrix approach (see e.g. \cite{galitskii1958energy,abrikosov1958concerning,Engelbrecht1992,arovas2022hubbard}) to calculate the effective interaction of electrons near the Fermi level, renormalized by virtual scattering to high-energy states. 
Physically, this approach amounts to solving exactly the two-particle scattering problem, and using the scattering amplitude as an effective interaction for particles near the Fermi level.
The effective interaction is given by an infinite sum of ladder diagrams, shown in Fig. \ref{fig:Tmatrix}. 
In two dimensions, the calculation is controlled by the parameter $1/\log(W/E_F)$~\cite{Engelbrecht1992}, since diagrams not included in Fig.~\ref{fig:Tmatrix} are suppressed by a factor of $\log(W/E_F)$ relative to the corresponding term in the ladder sum of the same order in $g_0$.

The effective interaction between flavor $\a$ and $\a'$ electrons with $\a \neq \a'$ is given by: 
\begin{equation}
    \tilde{g}_{\a\a'} = \frac{g_0}{1+\frac{4\pi g_0}{m} \Gamma_{\a\a'}},
\end{equation}
with $\Gamma_{\a\a'}$ the particle-particle susceptibility evaluated at the Fermi energy (see App.~\ref{app:Tmatrix} for details).
The non-zero effective interactions are
\begin{equation}\label{eq:T-Matrix}
    \begin{split}
       \tilde{g}_{XY} \equiv  \tilde{g}_{\s X, \s'Y}&=  \frac{g_0}{1 +\frac{2g_0\log(W/\eF)}{\eta^{+1/2}+\eta^{-1/2}}};\\
      \tilde{g}_{XX}\equiv   \tilde{g}_{\ua X, \da X}&=  \frac{g_0}{1 +g_0 \log(W/\eF)}.
    \end{split}
\end{equation}
Since $\frac{\sqrt{\h}+1/\sqrt{\h}}{2} >1$, 
we find that ${\tilde{g}}_{XX}< {\tilde{g}}_{XY}$. In terms of $\tilde{g}_{\alpha\alpha'}$, the valley and spin susceptibilities are given by:
\begin{align}
    \chi_{\vv}&=\frac{2m}{\pi}\frac{1}{1+\tilde{g}_{XX}-2\tilde{g}_{XY}},\nonumber\\
    \chi_{\ss}&=\frac{2m}{\pi}\frac{1}{1-\tilde{g}_{XX}}.
    \label{eq:T_matrix_susceptibilities}
\end{align}
This implies that $\chi_{\vv}>\chi_{\ss}$.

{The result of this section relies crucially on the anisotropy of the quadratic part of the dispersion. For instance, we considered a different problem of two valleys with an isotropic mass tensor but opposite trigonal warping terms. In this case, the inter-valley and intra-valley particle-particle susceptibilities are equal to the leading order in $\log(W/E_{F})$. However, the inter-valley particle-particle susceptibility is larger to subleading order in $\log(W/E_F)$. See App.~\ref{app:TrigonalWarping} for details.   }

\section{Discussion}\label{sec:Discussion}

We have shown how correlation effects tend to stabilize a valley-polarized Ising nematic state in a multi-valley 2DEG. This effect is enabled by the different anisotropies of the electronic dispersions in the two valleys. Physically, this effect originates from the fact that a pair of electrons from the same valley can avoid each other more effectively than two electrons from different valleys. Equivalently, the interaction between electrons is more effectively screened when they all occupy the same valley. 
The system can thus gain potential energy by populating only a single valley. A similar effect has been predicted to lead to a rotation symmetry-broken state in a 2DEG with a dispersion that is minimal along a ring in momentum space~\cite{Yang2006,Gopalakrishnan2011,Berg2012Rashba2DEGcrystal,Ruhman2014}. Our analytical results, using either RPA or the T-matrix approximation, are in good qualitative agreement with recent variational Monte Carlo calculations \cite{valenti2023nematic}, and with experiments in AlAs 2DEGs \cite{Hossain2021AlAsValleyPolarization}.

It is interesting to discuss our results in the context of the recently discovered strongly correlated 2DEGs in multilayer graphene and transition metal dichalcogenides (TMDs). These systems have multiple valleys in their electronic structure, and show a variety of symmetry-broken metallic states, including spin and valley-polarized states. The symmetry-broken phases are often described within the Hartree-Fock approximation; our analysis shows that, at least in the case of AlAs, HF is inadequate to describe the system even qualitatively, since it does not include the crucial correlation effects. This is perhaps not surprising, since we are dealing with a strongly correlated system.

It would be interesting to extend our work to other multi-valley 2DEGs with different symmetries. A hexagonal system with Fermi pockets centered at the $M$ points should have similar physics to the one we have described in this work, since each pocket has an anisotropic effective mass tensor, as in the C$_4$ symmetric case. The situation may be different in the case where the pockets are centered at the $K$ and $K'$ points of the Brillouin zone, as in graphene and in TMDs, due to the different symmetry at these points. 

\begin{acknowledgments}
We thank Andrey Chubukov for useful discussions and for his comments on this manuscript. 
V.C., S.A.K and E.B. were funded, in part, by NSF-BSF award DMR-2310312.
E.B. acknowledges support from the European Research Council (ERC) under grant HQMAT (grant agreement No. 817799).
This work has received funding from the
European Research Council under grant agreement no. 771503. S.D.H. acknowledges support by the Benoziyo Endowment Fund for the Advancement of Science. The Flatiron Institute is a division of the Simons Foundation.
\end{acknowledgments}


\bibliography{References}
\appendix

\newpage
 \onecolumngrid

\section{Hartree-Fock analysis}\label{app:Hartree-Fock analysis}

\subsection{Setup}
To do the Hartree-Fock analysis it is more convenient to start with the second quantized version of the Hamiltonian in Eq.~\ref{eq:Hamiltonian}:
\begin{equation}\label{eq:Hamiltonian:2ndQ}
    H = \sum_\a\sum_{\kbs} \epsilon_{\a}(\kbs) c^{\dagger}_{\kbs,\a}c^{}_{\kbs,\a}  + \frac{1}{2\A}\sum_{\a,\a'}\sum_{\kbs_1,\kbs_2; \qbs\neq \zero} \tilde{v}(\qbs) c^{\dagger}_{\kbs_1+\qbs,\a}
    c^{\dagger}_{\kbs_2-\qbs,\a'}
    c^{}_{\kbs_2,\a'}
    c^{}_{\kbs_1,\a}
\end{equation}
with $\epsilon_{\a}(\kbs) = \frac{k_x^2/\sqrt{\eta_\a}+k_y^2 \sqrt{\eta_\a}}{2m}$. We have excluded the $\qbs=0$ component {due to the neutralizing background}. Written in this way, it is obvious that the system has a $\U(2)\times \U(2)$ internal symmetry with generators 
\begin{equation}\label{eq:GeneratorsOfContinuousSymmetries}
    Q_{\t}^{a} = \sum_{\kbs} c^{\dagger}_{\kbs,(\sigma,\t)} s^{a}_{\s\s'}c_{\kbs,(\sigma,\t)}; \quad a=0,1,2,3; \tau = X, Y;
\end{equation}
where $s^0$ is the identity matrix and $s^a$ are the Pauli matrices. 

In addition to the above internal symmetries, we have $\RR^2\rtimes D_4$ internal-spatial symmetries - $\RR^2$ corresponds to translations and $D_4$ is the point group action such that $\pi/2$ rotations also sends $X\leftrightarrow Y$. The total symmetry group is $(\RR^2\times \U(2)\times \U(2) )\rtimes D_4$.

{As discussed in the main text, we will focus on states that are translation invariant, which means we will not consider charge-density waves or spin-density waves. In App.~\ref{app:diagonal}, we consider states that are obtained by occupying eigenstates of the kinetic energy operator. In App.~\ref{app:IVCstates}, we consider inter-valley coherent states.}

\subsection{`Diagonal' states}\label{app:diagonal}
Our variational states are specified by the electron density for each flavor, $n_{\alpha}$, and mass anisotropy used to define the Fermi surface, $\bar{\eta}_\a$. As we are interested in comparing the energies of states at the same electron density $\nel$, we can parametrize our states by the flavour filling $\nu_\a \equiv \frac{n_\a}{\nel}$. These are constrained to $\sum_{\a}\nu_\a =1$. To be more explicit, our variational states are 
\begin{equation}
    \ket{\nu_{\a},\tilde{\h}_\a} = \prod_{\a} \left( \prod_{\kbs: k_x^2/ \sqrt{\bar{\h}_\a}+ k_y^2 \sqrt{\bar{\h}_\a}\leq 4\p n_{\a} }c^{\dagger}_{\kbs,\a}\right) \ket{\text{Vaccum}}.
\end{equation}
It will be convenient to use express $n_{\a}$ in terms of energies via $\bar{E}_{F,\a} \equiv \frac{2\p n_{\a}}{m}$.

The energy per unit area $\Emc$, $\Emc\equiv \expval{H}/\A$, on these states are equal to $\Emc = \sum_{\a} \Emc_{\a}$, where the contribution of each flavor, $\Emc_{\a}$, is
\begin{equation}\label{eq:EnergyAlpha}
     \Emc_{\a}=\frac{m\bar{E}_{F,\a}}{2\pi}*\Bigg(
     \left(\sqrt{\frac{\h_\a}{\bar{\h}_{\a}}}+\sqrt{\frac{\bar{\h}_{\a}}{\h_\a}}\right)\frac{\bar{E}_{F,\a}}{4} - \frac{8}{3\pi}\frac{\Ksc(\bar{\eta}_{\a})}{m\aB} \sqrt{\frac{m \bar{E}_{F,\a}}{2}}
     \Bigg),
\end{equation}
where $\Ksc(\eta) = \frac{2\h^{1/4}\Kell(1-\h}{\p}$  and $\Kell$ is the elliptic K function.

Notice that $\Emc_{\a}$ is written in the form $n*(K-V_x)$ where $n$ is the fermion density, $K$ is the kinetic energy, and $-V_x$ is the exchange constribution. The kinetic energy is minimized at $\bar{\eta}_{\a} = \eta_\a$, while the exchange energy is minimized at $\bar{\eta}_{\a}=1$. Therefore, the optimal anisotropy, $\bar{\eta}_{\l}^*$, satisfies $0 < \abs{\log(\bar{\eta}_{\a}^*)} < \abs{\log(\eta_\a)}$. $\bar{\eta}_{\a}^*$ satisfies the equation 
\begin{equation}
    \frac{m\bar{E}_{F,\a}^2}{4\pi}*\Bigg(
    \frac{1}{4\bar{\eta}_\a^*}\left(-\sqrt{\frac{\eta_\a}{\bar{\eta}_{\a}^*}}+\sqrt{\frac{\eta_\a}{\bar{\eta}_{\a}^*}}\right) 
   - \frac{8\Ksc'(\bar{\eta}_{\a}^*)}{3\pi}\sqrt{\frac{2m }{m\aB^2\bar{E}_{F,\a}}}
    \Bigg) =0. 
\end{equation}
Therefore, we can write
\begin{equation}\label{eq:OptimalEta}
    \bar{\eta}_\a^* = \Fmc_1(\eta_\a, m\aB^2 \bar{E}_{F,\a}/2)
\end{equation}
for some function $\Fmc_1$ such that $\lim_{x\to 0} \Fmc_1(\eta,x) = 1$ and $\lim_{x\to +\infty} \Fmc_1(\eta,x) = \eta$. 

Thus, the optimal $\Emc_\a^*$ with respect to $\bar{\eta}_\a$ can be written as 
\begin{equation}
    \Emc_\a^* = \frac{m \bar{E}_{F,\a}^2}{4\pi} * \Fmc_2(\eta_\a, 1/\sqrt{m\aB^2  \bar{E}_{F,\a}/2} )
\end{equation}
with $\Fmc_2$ obtained by replacing $\bar{\eta}_{\a}$ of Eq.~\ref{eq:OptimalEta} in Eq.~\ref{eq:EnergyAlpha}. From the symmetry $(\eta_\a,\bar{\eta}_\a) \to (1/\eta_\a,1/\bar{\eta}_\a)$ in Eq.~\ref{eq:EnergyAlpha}, we obtain that $\Fmc_2(x,y) = \Fmc_2(1/x,y)$. Using that $\eta_\a =\eta^{\pm 1}$, we can write 
\begin{equation}\label{eq:EHF_functional}
    \Emc[\vec{\nu}; \rs] =\frac{\pi \nel^2}{m} \sum_\a \nu_\a^2\Fmc_2(\eta, \rs/\sqrt{\nu_\a}),
\end{equation}
where $\bar{E}_{F,\a} = \frac{2\pi \nu_\a}{m}\nel$. The $\nu_\a$'s are constrained to satisfy $\sum_\a \nu_\a =1$ and $\nu_a\geq 0$. Note that Eq.~\ref{eq:EHF_functional} is symmetric under permutations of $\nu_\a$. We can thus restrict to the subspace $1\geq \nu_1 \geq \nu_2 \geq \nu_3 \geq \nu_4 \geq 0$.

To find the minima of $\Emc$ in Eq.~\ref{eq:EHF_functional}, we systematically explored the parameter space $(\rs,\log_2(\eta))$ within the range $(\rs,\eta)$. Sampling with step sizes $(1/2^8,0)$ and $(0,1)$, we computed $E$ for each $(\rs,\h)$ pair on a grid of $\vec{\nu}$ values. The grid for $\vec{\nu}$ was constructed using three variables $(u_1,u_2,u_3)$ on the unit cube using the raltion $\nu_0 = 1/(1+u_1(1+u_2(1+u_3)))$ and $\nu_{j+1}= u_j \nu_j$ for $j=1,2,3$. We sampled $(u_1,u_2,u_3)$ both of a uniform grid with steps $1/8$ and with a uniform random distribution in the unit cube. We only found global minima of the form
\begin{equation}\label{eq:nuHF_solutions}
    \nu_{\a}^{*}= \begin{cases}
        \frac{1}{n} ;\quad &\a \leq n;\\
        0 ;\quad &\a> n.\\
    \end{cases}
\end{equation}
where $n=1,2,3,4$. 

Going back to the unconstrained $\nu_{\a}$ space, we find $\frac{4!}{n! (4-n)!}$ equivalent solutions for each solution in Eq.~\ref{eq:nuHF_solutions}. 
\begin{itemize}
    \item The $n=4$ solution simply corresponds to the symmetric liquid state. 
    \item The states $n=3$ are four-fold degenerate which corresponds to choosing one of the four flavors. This state will break the $\U(2)\times\U(2)$ internal symmetry down to a $\U(2)\times \U(1)^2$ subgroup. These states are nematic because the two valleys are inequivalent. 
    \item There are 6 states with $n=2$. Two of these states are the valley-polarized states and only break $C_4$ rotation symmetry. The other four states correspond to occupying different valleys with parallel or anti-parallel spins. 
    The former situation corresponds to `ferromagnetic' states, while the latter corresponds to `altermagnetic' states, {a.k.a.} `nematic spin nematic' states \cite{Oganesyan2001nematicFermiLiquid,Review2003DetectFluctuatingStripes}. 
    The `ferromagnetic' states break the internal symmetry down to $\U(1)^2\times\U(1)^2$ and they are related to the other `ferromagnetic' state and the `altermagnetic' state by some element in $\U(2)\times\U(2)$.
    \item The $n=1$ states break the same symmetries as the $n=3$ states.
\end{itemize}

We can obtain more degenerate states by using the symmetries of the full model. In Tab.~\ref{tab:table_SymmetryBreaking_HF}, we summarize the symmetry breaking patterns corresponding to the $n$-flavor states. 

\begin{table}
    \centering
    \begin{tabular}{c|c|c|c|c|c}
        $n$ & $1$ & $2_v$ & $2_s$ & $3$ & $4$\\\hline
         $H$& $\U(2)\times\U(1)^2\times D_2$  & $\U(2)\times\U(2)\times D_2$ & $\U(1)^2\times\U(1)^2\rtimes D_4$  & $\U(2)\times\U(1)^2\times D_2$ & $\U(2)\times\U(2)\rtimes D_4$ \\
         $\Mmc$& 
         $S^2\sqcup S^2$  & $*\sqcup *$ & $S^2\times S^2$ & $S^2\sqcup S^2$  & $*$ \\
    \end{tabular}
    \caption{Summary of symmetry-breaking patterns for the $n$-flavor states obtained in Hartree-Fock. $H$ is a group such that $\RR^2\rtimes H$ is isomorphic to the unbroken symmetries of the trial state. $\Mmc$ is the order parameter manifold. We have separated the $n=2$ states into two groups $2_v$ and $2_s$ which are not related by any symmetries of the many-body Hamiltonian. $S^2$ is the 2-dimensional sphere, $\sqcup$ denotes disjoint unit, and  $*$ is a single point. }
    \label{tab:table_SymmetryBreaking_HF}
\end{table}

\subsection{Inter-valley coherent states}\label{app:IVCstates}

We will focus on the competition between inter-valley coherent states and the `diagonal' fluid states. 

\subsubsection{Non-interacting susceptibilities}\label{app:IVC:weak-coupling}
\def\EF{E_{\rm{F}}}

Here, we calculate the non-interacting susceptibility for IVC ordering and compare it to the susceptibility for spin or valley ordering. 
For a contact (momentum-independent) interaction, the stability of the symmetric phase can be analyzed using the Stoner criterion, {$g \chi  > 1$}, where $g$ is the interaction strength and $\chi$ the corresponding susceptibility. Thus, the instability towards the ordering with the larger susceptibility onsets first. 

The non-interacting (uniform) IVC susceptibility is
\begin{equation}
    \begin{split}
        \chi_{\rm{IVC}} 
        &= -4\int\frac{\dd^2\kbs}{(2\p)^2} \frac{n_{F}(\varepsilon_{X,\kbs}) - n_{F}(\varepsilon_{Y,\kbs})}{\varepsilon_{X,\kbs}-\varepsilon_{Y,\kbs} + \ii 0^+},
    \end{split}
\end{equation}
where $n_F(\varepsilon) = \Theta(\EF - \varepsilon)$, $\varepsilon_{X,\kbs} = (k_x^2 / \sqrt{\eta} + k_y^2 \sqrt{\eta} )/ 2m$ and $\varepsilon_{Y,\kbs} = (k_y^2 / \sqrt{\eta} + k_x^2 \sqrt{\eta} )/ 2m$. 

After some algebra, we obtained
\begin{equation}
    \chi_{\rm{IVC}} = 4\frac{m}{2\pi}  \frac{2}{\sqrt{\h}-1/\sqrt{\h}} \arcsin(\frac{\sqrt{\h}-1/\sqrt{\h}}{\sqrt{\h}+1/\sqrt{\h}}) \leq 4\frac{m}{2\pi} = \chi_{\rm{v}}. 
\end{equation}
The last inequality follows by writing $\h =\ee^{2\s}$ so that $\frac{\chi_{\rm{IVC}}}{4\frac{m}{2\pi}} =  \frac{\arcsin(\tanh\s))}{\sinh(\s)}$. Plotting the later function, one can check that it is peaked at $\s=0$. 
Thus, we find the ellipticity of the Fermi surface ($\eta \ne 0$) enhances the valley (or spin) susceptibility relative to the IVC susceptibility. 

\subsubsection{Strong coupling limit of Hartree-Fock}\label{app:IVC:strong-coupling}

{We compare the energy of the spin-valley polarized (SVP) state with that of an inter-valley coherent (IVC) state that preserves $C_4$ symmetry. Even though at strong coupling (large $\rs$), the correct phase is a Wigner crystal, we can still ask what is the asymptotic behaviour of the Hartree-Fock energy for the IVC and SVP states. As $\rs$ is large, we can minimize the exchange by occupying at most one mode for every momenta. }

\textbf{Spin-valley polarized state:} \textit{We assume that the optimal shape of the Fermi surface is an ellipse with anisotropy $\tilde{\h}=\ee^{2\xi}>1$.} In this case, we can evaluate $E$ explicitly:
\begin{equation}
    \frac{E}{\A} = \frac{\kF^4}{16\p m} \cosh(\s-\xi) - \Ksc(\ee^{2\xi}) \frac{e^2}{4\pi\epsilon}\frac{\kF^3}{3\p^2}
\end{equation}
where $\eta = \ee^{2\s}$ and $\Ksc(x)$ is defined in Eq.~\ref{eq:def:Ksc}.
In (effective) Hartree units, we obtain 
\begin{equation}\label{eq:Enot:StrongCoupling}
    \frac{E}{\Nel } = \frac{\cosh(\s-\xi)}{\rs^2} - \frac{8\Ksc(\ee^{2\xi})}{3\p }\frac{1}{\rs}.
\end{equation}
At large $\rs$, $0< \m \ll 1$, so we can expand the energy around $\xi=0$:
\begin{equation}
    \begin{split}
        \frac{E}{\Nel} &\approx  \frac{E_0}{\Nel} - \frac{\sinh(\s) \xi}{\rs^2} +  \frac{\xi^2}{6\p \rs};\\
        \frac{E_0}{\Nel}&= \frac{\cosh(\s)}{\rs^2} - \frac{8}{3\p \rs},
    \end{split}
\end{equation}
where $\Nel = \A \nel$ and $\Ksc(\ee^{2\xi})\approx 1-\frac{\xi^2}{16}$ for $ \abs{\xi}\ll1$. Minimizing with respect to $\xi$, we obtain 
\begin{equation}
    \xi = \frac{3\p }{\rs} \sinh(\s);\quad \frac{E}{\Nel} = \frac{E_0}{\Nel} - \frac{\sinh[2](\s)}{2/(3\p)}\frac{1}{\rs^3} 
\end{equation}
to leading order in $\rs$.

\textbf{Inter-valley coherent states:} We occupy the modes created by $f^{\dagger}_{\kbs} = \sum_{\a} c^{\dagger}_{\kbs,\alpha}u_{\kbs,\a} $, where $u_{\kbs,\a} = \delta_{\a,1}  \a_{\kbs}+ \delta_{\a,3} \b_{\kbs}$ with $\abs{\a_{\kbs}}^2+\abs{\b_{\kbs}}^2=1$. We consider the following ansatz $\a_{\kbs} = \frac{1+\d_{\kbs}} {\sqrt{2(1+\abs{\d}_{\kbs}^2)}}$ and $\b_{\kbs} = \frac{1-\d_{\kbs}} {\sqrt{2(1+\abs{\d}_{\kbs}^2)}}$ with $\delta_{\kbs}$ small. This state preserves $C_4$ symmetry as long as $\delta_{C_4\kbs} = - \delta_{\kbs}$.

The kinetic energy of the $f^\dagger_{\kbs}$ mode is 
\begin{equation}
    \bar{\eps}(\kbs)=\frac{ k^2}{2m} \cosh(\s) -\frac{ k^2 \cos(2\q) }{2m}\sinh(\s) \frac{2\Re[\d_{\kbs}]}{(1+\abs{\d_{\kbs}^2})}.
\end{equation}
where $k_x +\ii k_y  = k\ee^{\ii\q} $.

The exchange energy per area is  
\begin{equation}
    \begin{split}
     \bar{E}_{x}&=- \frac{1}{2}\int_{\kbs_1,\kbs_2}\tilde{v}(\kbs_1-\kbs_2) \Fmc(\kbs_2,\kbs_1);
\\
        \Fmc(\kbs_2,\kbs_1) 
        &= 1 - \frac{\abs{\delta_{\kbs_1}-\delta_{\kbs_2}}^2}{(1+\abs{\delta_{\kbs_1}}^2)(1+\abs{\delta_{\kbs_2}}^2)}.
    \end{split}
\end{equation}
where $\int_{\kbs} \equiv \int \frac{\dd^2{\kbs}}{(2\p)^2}$ is integrated over the filled states.

\textit{We now assume that the Fermi surface is circular.} As we expect $\delta_{\kbs}$ to be small and of order $1/\rs$, we expand the kinetic energy to linear order in $\delta$ and the exchange energy to order $\d^2$:
\begin{equation}
    \frac{E-E_0}{\A}= - \int_{\kbs: \abs{\kbs}\leq \kF} \frac{\cos(2\q)\kbs^2 \sinh(\s_0)}{2m} \left(2\Re[\d_{\kbs}]\right) + \frac{1}{2}
    \int_{\kbs_1,\kbs_2: \abs{\kbs}_{1(2)}\leq \kF}\tilde{v}(\kbs_1-\kbs_2) \abs{\delta_{\kbs_1}-\delta_{\kbs_2}}^2,
\end{equation}
where $E_0$ is the same as in Eq.~\ref{eq:Enot:StrongCoupling}.
We see that to minimize $E$, $\delta_{\kbs}$ should be real.

We can optimize the energy functional with respect to $\delta_{\kbs}$ because it is quadractic: 
\begin{align}
        \Emc[\delta_{\kbs}] &\equiv - \int_{\kbs: \abs{\kbs}\leq \kF} \frac{\cos(2\q)\kbs^2 \sinh(\s_0)}{2m} \left(2\d_{\kbs}\right) + \frac{1}{2}
    \int_{\kbs_1,\kbs_2: \abs{\kbs}_{1(2)}\leq \kF}\tilde{v}(\kbs_1-\kbs_2) ({\delta_{\kbs_1}-\delta_{\kbs_2}})^2,\label{eq:EnergyFxnl}\\
    \frac{\var \Emc[\delta_{\kbs}] }{\var \delta_{\kbs}} &=-\frac{\cos(2\q)\kbs^2 \sinh(\s_0)}{2m} \left(2\right) +
    2\int_{\pbs: \abs{\pbs}\leq \kF}\tilde{v}(\pbs-\kbs) ({\delta_{\kbs}-\delta_{\pbs}})=0,\label{eq:Optimal}
\end{align}
where $\delta_{\kbs}$ is a function for $\kbs$ in the region $\abs{\kbs}<\kF$. The second factor of $2$ in the exchange is obtained by using $\tilde{v}(\qbs)= \tilde{v}(-\qbs)$. 

We solved Eq.~\ref{eq:Optimal} as follows. We decompose $\delta_{\kbs}$ in angular harmonics and found that due to the rotation invariance of the Coulomb interaction, only the $\cos(2\q)$ component is non-zero. After performing the angular integrations, we can solve for the $k$-dependence of $\delta_{\kbs}$ by iteration. We found that the coefficient of $-\sinh^2(\s_0)/\rs^3$ is approximately equal to $4.122$ which is smaller than $3\p/2\approx 4.712$ for the SVP state. Thus, the SVP state is favored compared to the IVC state.

\paragraph{Deformation of Fermi surface is subleading:}

In the IVC state, let's calculate the leading correction to the kinetic energy upon deforming the Fermi surface. We can do this by setting 
\begin{equation}
    \tilde{k}^2_{\rm{F}}(\q) = \kF^2[1  + \alpha(\q)]
\end{equation}
and $\alpha$ small. The density imposes the condition that
\begin{equation}
    \int_0^{2\p}\int_0^{\tilde{k}(\q)} \frac{k\dd{k}}{2\pi}\frac{\dd{\q}}{2\pi} = \frac{\kF^2}{4\pi} \Longrightarrow \int_0^{2\pi}\alpha(\q) = 0.
\end{equation}

Now the kinetic energy (ignoring the IVC canting) is 
\begin{equation}
    \int_0^{2\p}\int_0^{\tilde{k}(\q)} \frac{\cosh(\sigma)k^2}{2m}\frac{k\dd{k}}{2\pi}\frac{\dd{\q}}{2\pi} = 
    \frac{\cosh(\sigma)}{16 \pi m} \int_0^{2\p} [1+\alpha(\q)]^2\dd{\q}=  
    \frac{\cosh(\sigma)}{16 \pi m} \int_0^{2\p} [1+\alpha(\q)^2]\dd{\q} 
\end{equation}
We thus see to this order the kinetic energy is minimized by $\alpha=0$. As the exchange energy is already a minimum, there is no energy gain to leading order for the IVC. In the presence of canting, the leading correction would be to take $\alpha(\q)  = \alpha_4\cos(4\q) $ but $\alpha_4$ will be of higher order than the canting itself.

\section{Energy of the anisotropic 2-valley 2DEG}\label{app:Energy2DEGAnisotropic}

In this appendix, we will work in units of the effective Hartree $\mathrm{Ha}_* := \frac{\hbar^2}{m \aB^2}$ where $\aB = \left(\frac{m e^2}{4\pi \epsilon \hbar^2}\right)^{-1}$. In this appendix, we report energies per electron to be aligned with the 2DEG literature. 

As mentioned in the main text, our model in Eq.~\ref{eq:Hamiltonian} conserves the number of particles of each flavor $\alpha$, in addition to a $\SU(2)$-spin rotation symmetry for each valley flavor. One can ask what is the ground state in various polarization sectors. For concreteness, we will focus on the states in Tab.~\ref{tab:DefinitionStates}. States 1 through 4 are `fully-polarized' states while states 5 and 6 are `partially-polarized' states. In the latter states, the sectors are parametrized by a number $\zeta_{v} = \frac{n_X-n_Y}{n_X+n_Y}$ and $\zeta_{s} = \frac{n_\ua - n_\da}{n_\ua+n_\da}$. {In the main text, we call the `fully-polarized' states `polarized' for ease of notation.}

\begin{table}[t]
    \centering
    \begin{equation*}
        \begin{array}{|c|c|c|c|c|c|}\hline
       \# &\text{State} & \nu_{\ua X} & \nu_{\da X}& \nu_{\ua Y}& \nu_{\da Y} \\  \hline
       1 & \text{Symmetric (Sym)} & 1/4& 1/4& 1/4& 1/4 \\
       2 & \text{ValleyPol (VP)} & 1/2& 1/2& 0& 0 \\
       3 & \text{SpinPol (SP)} & 1/2 & 0 &1/2 &0 \\
       4 & \text{SpinValleyPol (SVP)} & 1 & 0 & 0 &0 \\\hline
       5 & \text{PartialValleyPol (pVP)} & \frac{1+\zeta_v}{4} & \frac{1+\zeta_v}{4} & \frac{1-\zeta_v}{4} &\frac{1-\zeta_v}{4}  \\
       6 & \text{PartialSpinPol (pSP)} & \frac{1+\zeta_s}{4} & \frac{1-\zeta_s}{4} & \frac{1+\zeta_s}{4} &\frac{1-\zeta_s}{4}  \\\hline
    \end{array}
    \end{equation*}
    \caption{\label{tab:DefinitionStates} Flavor occupancies of the various states we analyze, where   $\nu_{\sigma v}$ is the fraction of the electrons with spin $\sigma=\ua,\da$ and valley $v=X,Y$. $\zeta_{v/s} \in[0,1]$ are the partial polarization parameters. }
\end{table}

On each sector, the ground-state energy per electron admits an expansion 
\begin{align}
    E  &=  E_{\kk} + E_{\xx} + E_{\cc}  \\
    E_{\text{k}} &= \frac{\varepsilon_{\kk}}{\rs^2}\\
    E_{\xx} &= \frac{\varepsilon_{\xx}}{\rs}\\
    E_{\cc} &= \varepsilon_{0,\rr}+ \varepsilon_{0,\xx} + \varepsilon_{0,\xx'} + \varepsilon_{1,\ell}\rs \log(\rs) + \Omc(\rs)
\end{align}
where $\varepsilon_a$ are dimensionless numbers that depend on $\eta$ and the sector (polarization pattern). $E_k$ is the kinetic energy, $E_{\text{x}}$ is the exchange energy, and $E_c$ is the correlation energy. 

We calculated the coefficients $\varepsilon_{\kk}, \varepsilon_{\xx}, \varepsilon_{0,\rr}$, $\varepsilon_{0,\xx}$, $\varepsilon_{0,\xx'}$ and $\varepsilon_{1,\ell}$ for the first four sectors in Tab.~\ref{tab:DefinitionStates}. The results can be found in Tab.~\ref{tab:summaryExpansionParams}. The calculations are shown in the following appendices. $\varepsilon_{\kk}$ comes from the non-interacting kinetic energy in each sector, $\varepsilon_{\xx}$ corresponds to the exchange energy. $\varepsilon_{0,\rr}$ and $\varepsilon_{0,\xx}$ come from standard second-order perturbation theory (diagrams \textbf{b} and \textbf{c} in Fig.~\ref{fig:2PT}). Diagram \textbf{d} vanishes identically at zero temperature due to the pole structure in the $\omega$-plane of the propagators. However, according to Kohn and Luttinger \cite{KohnLuttinger1960}, standard perturbation theory breaks down already at second order in the interactions. Instead, one ought to take the thermodynamic limit before sending the temperature to zero. Doing this gives an additional contribution $\varepsilon_{0,\xx'}$ to the correlation energy that is a combination of effects from diagrams \textbf{a} and \textbf{d} in Fig.~\ref{fig:2PT}.
\begin{table*}[t]
    \centering
    \begin{equation*}
    \begin{array}{|c||c|c|c|c|}\hline
          & \text{Sym} & \text{VP} & \text{SP} & \text{SVP}\\\hline
          \varepsilon_{k} & \frac{1}{4} & \frac{1}{2} & \frac{1}{2} & 1 \\
          \varepsilon_{x} & -\frac{4}{3\pi}\Ksc(\eta) & -\frac{4\sqrt{2}}{3\pi}\Ksc(\eta) & -\frac{4\sqrt{2}}{3\pi}\Ksc(\eta) & -\frac{8}{3\pi}\Ksc(\eta) \\
          \varepsilon_{0,r} & (\log(2)-1)(1+\Dsc(\eta)) &  \log(2)-1 & \frac{(\log(2)-1)}{2 }(1+\Dsc(\eta)) & \frac{(\log(2)-1)}{2} \\
          \varepsilon_{0,x} & \Xsc(\eta)& \Xsc(\eta)& \Xsc(\eta)& \Xsc(\eta) \\
        \varepsilon_{0,x'} & \Vsc(\eta)& \Vsc(\eta)& \Vsc(\eta)& \Vsc(\eta) \\
          \varepsilon_{1,\ell} &  2\left(1-\frac{10}{3\pi}\right)(\Esc(\h)+3\Fsc(\h))  & 
          \sqrt{2}\left(1-\frac{10}{3\pi}\right)\Esc(\h) & \frac{\sqrt{2}}{4}\left(1-\frac{10}{3\pi}\right)(\Esc(\h)+3\Fsc(\h))  &  \frac{1}{4}\left(1-\frac{10}{3\pi}\right)\Esc(\h) \\\hline
    \end{array}
\end{equation*}
    \caption{Summary of the expansion parameters for the first four states in Tab.~\ref{tab:DefinitionStates}. The functions $\Ksc, \Dsc, \Xsc, \Esc, \Fsc$ and $\Vsc$ are defined so that they are positive and are invariant under $\eta \to 1/\eta$. See Eqs.~\ref{eq:def:Ksc},~\ref{eq:DefDsc},~\ref{eq:ExchangeEnergy2PT},~\ref{eq:DefEandF}, \ref{eq:Vsc:def} in the main text, respectively.}
    \label{tab:summaryExpansionParams}
\end{table*}

\subsection{First-order perturbation results}
The calculation of $\varepsilon_k$ and $\varepsilon_{x}$ are standard results \cite{giuliani2008quantum}, as well as the expression for the polarization bubble (a.k.a. the Lindhard function). 

$\varepsilon_{\kk}/\rs^2$ is equal to the non-interacting energy and is equal in general to 
\begin{equation}
    \varepsilon_{\kk}  = \sum_{\alpha} \nu_{\alpha}^2.
\end{equation}
The exchange contribution is 
\begin{equation}
    {\varepsilon_{\xx}} = -\frac{8 \Ksc(\eta)}{3\pi}\sum_{\alpha} \nu^{3/2}_{\alpha},
\end{equation}
where 
\begin{equation}\label{eq:def:Ksc}
    \begin{split}
         \Ksc(\eta) 
         &= \int_{0}^{2\pi}\frac{\dd{\theta}}{2\pi} \frac{1}{\sqrt{\sqrt{\eta}\cos(\theta)^2+ \sin(\theta)^2/\sqrt{\eta}}} \\
         & = \frac{2 \eta^{1/4}\Kell(1-\eta)}{\pi}.
    \end{split}
\end{equation}
$\Kell$ is the elliptic K function. 

The Lindhard function (polarization bubble) of Slater determinant state with densities $\{n_{\a}\}_{\a=1,2,3,4}$ is given by 
\begin{equation}\label{eq:Lindhard}
    \Pi(\ii\Omega,\qbs) = \frac{m}{2\p}\sum_{\a}\left(1- \Re[\sqrt{\left(1+\frac{\ii\Omega}{\eps_{\a}(\qbs)}\right)^2- \frac{8\pi n_{\a}}{m\eps_{\a}(\qbs)}}]\right)
\end{equation}
where $\eps_{\a}(\qbs) = \frac{q_x^2/\sqrt{\h_\a}+q_y^2\sqrt{\h_\a}}{2m}$ is the dispersion. For example, for the symmetric state $n_{\a} = \nel/4$ while for the VP state, we have $n_{1} = n_{2} = \nel /2$ and $n_3=n_4=0$.

It is convenient to express the Lindhard function for a single electron flavor with Fermi momentum $\kF$, mass $m$, and anisotropy $\eta_{\a}$, as follows
\begin{equation}\label{eq:defPi}
    \begin{split}
    \Pi_\a(\ii\Omega,\qbs)  = \frac{m}{2\pi}\left(1-\frac{\Re(\sqrt{z^2-1})}{\Re(z)}\right);\\
        z = \frac{q}{2\kF} I(\eta_\a,\theta) + \frac{\ii \Omega}{\frac{\kF}{m}q I(\eta_\a,\theta)};\\
        I(\eta,\theta) = \sqrt{ \cos(\q)^2/\sqrt{\eta} + \sin(\q)^2\sqrt{\eta}},
    \end{split} 
\end{equation}
where $\qbs = (q \cos(\q), q \sin(\q))$.

\subsection{Second-order ring diagram (\texorpdfstring{$\varepsilon_{0,r}$}{e})}\label{app:Varepsilon:0r}

Diagram \textbf{b} in Fig.~\ref{fig:2PT} can be separated into contributions for pairs of flavors $(\alpha,\alpha')$  as 
\begin{equation}
    \varepsilon_{0,\text{r}} = \sum_{\a,\a'} \varepsilon_{0,\text{r}}(\a,\a'),
\end{equation}
where
\begin{equation}\label{eq:TwoLoopIntegral}
    \varepsilon_{0,\text{r}}(\a,\a') = -\frac{1}{2\nel}\int_{0}^{\infty}\!\!\frac{\dd{\Omega}}{2\pi} \int\!\!\!\frac{\dd^2{\qbs}}{(2\pi)^2} \left[\tilde{v}(\qbs)\right]^2 \Pi_{\a}(\ii\Omega,\qbs)\Pi_{\a'}(\ii\Omega,\qbs)
\end{equation}
Here $\Pi_{\a}(\ii\Omega,\qbs)$ is the polarization bubble for electrons of flavor $\alpha$ {(that depends on the density of electrons in that flavor)}, and $\tilde{v}(\qbs) = \frac{1}{m \aB} \frac{2\pi}{q}$.

To perform the integral in Eq.~\ref{eq:TwoLoopIntegral}, it is convenient to change to the variables $x= \frac{q}{2\kF}$, $y = \frac{\Omega}{\frac{\kF}{m} q}$, where $\kF$ is defined by $\nel =\frac{\kF^2}{4\pi}$. In these new variables, 
\begin{equation}
    \begin{split}
        \int_{0}^{\infty}\!\!\frac{\dd{\Omega}}{2\pi} \int\!\!\!\frac{\dd^2{\qbs}}{(2\pi)^2} &\to \text{Ha}_* \frac{2\nel}{\pi} \left(\frac{4}{\rs}\right)^2 \int_0^{\infty}\!\!\!\dd{x}\int_0^{\infty}\!\!\!\dd{y}\int_{0}^{2\pi}\!\!\frac{\dd{\q}}{2\pi} x^2\\
        \tilde{v}(\qbs) &=  \frac{\rs}{4} \frac{2\pi}{m} \frac{1}{x};\\
        \Pi_\a(\ii\Omega,\qbs)&= \frac{m}{2\pi}\Psi_\a(x,y,\theta);\\
\Psi_{\a}(x,y,\theta)& =  1-\frac{\Re(\sqrt{z_\a^2-\nu_\a})}{\Re(z_\a)}\\
    z_\a &= x I(\eta_\a,\theta) + \ii y / I(\eta_\a,\theta),
    \end{split}
\end{equation}
where $I(\eta,\theta)$ was defined in Eq.~\ref{eq:defPi}. We thus get
\begin{equation}\label{eq:TwoLoopIntegral:Ver2}
    \begin{split}
        \varepsilon_{0,\rr}(\a,\a') 
        &= -\frac{1}{\pi}
    \int_{0}^{\infty}\!\!\!\dd{x}
    \int_{0}^{\infty}\!\!\!\dd{y}
    \int_{0}^{2\pi}\!\!\frac{\dd{\q}}{2\pi} 
    \Psi_{\a}(x,y,\theta)\Psi_{\a'}(x,y,\theta).
\end{split}
\end{equation}
When $\h_\a =\h_{\a'}$, we can make the change of variables $x\to x/I_\a$ and $y\to y I_\a$, which has a trivial Jacobian. This removes the $\theta$ dependence on the $\Psi's$. $\varepsilon_{0,\rr}(\a,\a')$ is thus independent of $\eta$. In particular, it is equal to the isotropic case. The result for the latter case is known for arbitrary partial fillings \cite{Loos2011ExactDirect2nd}.

When $\h_\a = 1/\h_{\a'}=\eta$ and $\nu_\a = \nu_{\a'}=\nu$, we can make a further change of variables $(x,y) \to (\sqrt{\nu} x, \sqrt{\nu} y)$ to express $\varepsilon_{0,\text{r}}(\a,\a')$ in terms of a single function of $\eta$:
\begin{equation}\label{eq:DefDsc}
    \begin{split}
        \Dsc(\eta) &\equiv \frac{2}{\pi (1-\log(2)) }\int_{0}^{\infty}\!\dd{x}\int_{0}^{\infty}\!\dd{y}\int_{0}^{2\pi}\!\frac{\dd{\theta}}{2\pi} \left(1 - \frac{\Re[\sqrt{z_{X}^2-1}]}{\Re[z_{X}]}\right)
    \left(1 - \frac{\Re[\sqrt{z_{Y}^2-1}]}{\Re[z_{Y}]}\right);  \\
    z_A &= x I_A(\q) + \ii y /I_A(\q); \quad [A = X,Y]\\
    I_X(\q) &=\sqrt{\h^{-1/2}\cos(\q)^2+\h^{+1/2}\sin(\q)^2}\\
    I_Y(\q) &= \sqrt{\h^{+1/2}\cos(\q)^2+\h^{-1/2}\sin(\q)^2} 
    \end{split}
\end{equation}
$\Dsc$ satisfies $\Dsc(\eta) = \Dsc(1/\eta)$ and $\Dsc(1)=1$. We numerically evaluate $\Dsc(\eta)$ and show a plot of $\Dsc(\eta)$ in  Fig.~\ref{fig:Direct2Func}. We found $\Dsc(0) \approx 0.6813$. The leading ring contribution is given by  
\begin{equation}
    \varepsilon_{0,\text{r}}(\a,\a') = \begin{cases}
         \frac{\log(2)-1}{2} \nu_{\a}; \quad &\h_\a= \h_{\a'}\\
         \frac{\log(2)-1}{2} \nu_{\a} \Dsc(\h); \quad &\h_\a= 1/\h_{\a'}=\h\\
    \end{cases}
\end{equation}
where we have assumed that $\nu_\a=\nu_{\a'}$.

\begin{figure}[t]
    \centering
    \includegraphics[width=0.45\textwidth]{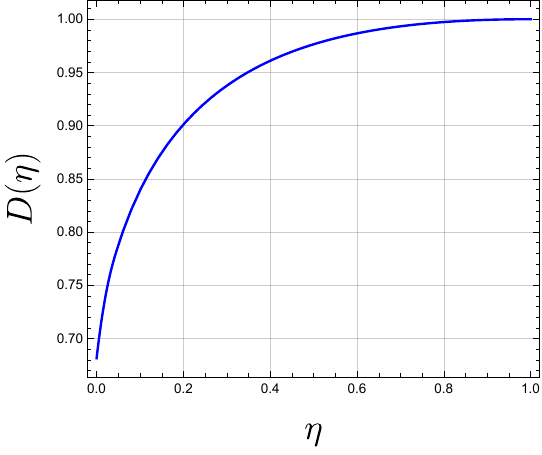}
    \caption{ Plot of $\Dsc(\eta)$ as defined in Eq.~\ref{eq:DefDsc}}
    \label{fig:Direct2Func}
\end{figure}

\onecolumngrid
\subsection{Second-order exchange diagram (\texorpdfstring{$\varepsilon_{0,x}$}{e0x})}

\def\thetaMeasure{\int_{0}^{2\pi}\!\!\frac{\dd{\theta}}{2\pi \abs{\ubs}_{W^{-1}}}}
We modify the calculation of Ref.~\cite{Ishihara1980ExactExchange2ndOrder} to the anisotropic case. We want to evaluate 
\begin{equation}
    \begin{split}
        \lim_{\beta \to \infty}-\frac{\log \Xi }{\beta A}&=\frac{1}{(2\pi)^6 }\int \tilde{v}(\qbs) \Lambda (\qbs) \dd{\qbs}; \\
        \Lambda(\qbs)&= \frac{1}{2}\int_{\RR}\frac{\dd{\Omega}}{2\pi} \dd{\pbs}\dd{\kbs}\tilde{v}(\pbs-\kbs) 
        \frac{\fmf(\pbs+\qbs)-\fmf(\pbs)}{\ii \W +\eps_{\qbs+\pbs}-\eps_{\pbs}}
        \frac{\fmf(\kbs+\qbs)-\fmf(\kbs)}{\ii \W +\eps_{\qbs+\kbs}-\eps_{\kbs}}.
    \end{split}
\end{equation}
$\tilde{v}(\qbs)$ is the potential, $\eps_{\qbs}= \qbs {W}\qbs /2m$ is the dispersion relation, ${W} = \diag(\sqrt{\eta},1/\sqrt{\eta})$. $\fmf(\qbs) = \Theta(\kF^2/(2m)-\eps_{\kbs})$ is the Fermi occupation function at zero temperature. It will be convenient to move the anisotropy from $\eps$ to $\tilde{v}$. We do this by sending $\qbs=(q_x,q_y) \to \sqrt{\bar{W}^{-1}}\qbs= (q_x/\h^{1/4},q_y\h^{1/4})$. The effect is that the potential now becomes $\tilde{v}'(\qbs)= \tilde{v}(\sqrt{\bar{W}^{-1}}\qbs)$ while reducing the dispersions to the isotropic case. From now on, let $m=1/2$ and $\kF=1$.

Going back to real space $\tilde{v}'(\qbs) = \int \dd{\rbs} v'(\rbs) \ee^{\ii\rbs\cdot(\qbs-\kbs)}$ ($u'(\rbs)={u}(\sqrt{W} \rbs)$), we define 
\begin{equation}
    \lambda(\ii\W,\qbs,\rbs)=\int \frac{\fmf(\pbs+\qbs)-\fmf(\pbs)}{\ii \W +\eps_{\qbs+\pbs}-\eps_{\pbs}}\ee^{\ii\pbs\cdot \rbs} \dd{\pbs}
\end{equation}
so that 
\begin{equation}
    \L(\qbs) = \frac{1}{2} \int_{\RR} \frac{\dd\Omega}{2\p} \int\dd{\rbs} v'(\rbs) \lambda(\ii\W,\qbs,+\rbs) \lambda(\ii\W,\qbs,-\rbs)
\end{equation}
Because the integrand in $\L(\qbs)$ is an even function of $\Omega$, we can restrict to $\W>0$. Then
\begin{equation}
    \l(\ii\W,\qbs,\rbs) = -4\pi \ee^{-\ii\qbs\cdot \rbs /2}\int_{0}^{\infty} \ee^{-\Omega \tau} \frac{J_1(\abs{\rbs+2\qbs \t})}{\abs{\rbs+2\qbs \t}} \sin(\frac{\qbs\cdot(\rbs+2\qbs \tau)}{2})\dd{\t} ; \W>0.
\end{equation}
where $J_1$ is a Bessel function of the first kind.

Plugging this in $\L(\qbs)$ and performing the $\W$ integral, gives us
\begin{equation}
    \L(\qbs) = 4\pi \int_{0}^{\infty}\int_{0}^{\infty}\frac{\dd{\t_1}\dd{\t_2}}{\t_1+\t_2}\int_{\RR^2} \dd{\rbs } v'(\rbs) 
    \frac{J_1(\abs{\rbs+2\qbs \t_1})}{\abs{\rbs+2\qbs \t_1}} \sin(\frac{\qbs\cdot(\rbs+2\qbs \t_1)}{2}) 
    \frac{J_1(\abs{-\rbs+2\qbs\t_2})}{\abs{-\rbs+2\qbs \t_2}} \sin(\frac{\qbs\cdot(-\rbs+2\qbs \t_2)}{2})
\end{equation}
In the original work, the authors already replaced $v'(\rbs) = e^2/r$ by this point.

We next multiply by $\tilde{v}'(\qbs)$ and integrate over $\qbs$. Due to the rotation invariance in the original work, the integration over $\theta_{\qbs}$ is not needed and only gives a factor of $2\pi$. In the present case, we can write $\qbs = q \ubs$ with $\ubs = [\cos(\q_{\qbs}),\sin(\q_{\qbs})]$, with $\ubs$ a unit vector. Replacing $\tau_j \to q \tau_j$ and performing the $q$ integral using Eq.~\ref{eq:Integral1} gives us \footnote{Note the  typo in the argument of the last logarithm in Equation 14 of Ref.~\cite{Ishihara1980ExactExchange2ndOrder}. }
\begin{equation}\label{eq:integralLu}
   \begin{split}
       \Imc &= \int_{\RR^2} \L(\qbs) \tilde{v}'(\qbs)\dd{\qbs} \\
       &= 2(2\pi)^2\int_{0}^{\infty}\int_{0}^{\infty}\frac{\dd{\t_1}\dd{\t_2}}{\t_1+\t_2}\int_{0}^{2\pi} \frac{\dd{\theta}}{2\pi} \tilde{v}'(\ubs) v'(\rbs) \frac{J_1(\abs{\rbs+2\ubs \t_1})}{\abs{\rbs+2\ubs \t_1}}  \frac{J_1(\abs{-\rbs+2\ubs\t_2})}{\abs{-\rbs+2\ubs \t_2}}  \log\abs{\frac{2(\t_1+\t_2)}{\ubs\cdot(2\rbs +2 \ubs\cdot(\t_1-\t _2))}}.
   \end{split}
\end{equation}
The difference with the original work is that $\tilde{v}'(\ubs) v'(\rbs) = \frac{2\pi e^4}{\sqrt{\ubs\cdot W^{-1}\cdot \ubs}\sqrt{\rbs\cdot W\cdot \rbs}} = \frac{2\pi e^4}{\abs{\ubs}_{W^{-1}}\abs{\rbs}_{W}}$ where we defined $\abs{\Rbs}_{A}\equiv \sqrt{\Rbs\cdot A \cdot \Rbs}$ for a 2 by 2 matrix $A$.

Next, we consider the change of variables $\rbs\to \rbs (\t_1+\t_2)+ \ubs(\t_2-\t_1)$ followed by $\tau_1 = \frac{1+\xi}{2}\sigma$ and $\tau_2 = \frac{1-\xi}{2}\sigma$ such that $-1<\xi<1 $ and $0<\s <\infty$ and $\dd{\t_1}\dd{\t_2}= \frac{\sigma}{2} \dd{\xi}\dd{\sigma}$ so that the original $\rbs$ goes to $\sigma(\rbs - \xi \ubs)$. Then
\begin{equation}
    \Imc=(2\pi)^3e^4\thetaMeasure \int_{\RR^2} \dd{\rbs} \int^{+1}_{-1} \frac{\dd{\xi}}{\abs{\rbs+ \xi \ubs}_{W}} \log\abs{\frac{1}{\rbs\cdot \ubs}} \int_{0}^{\infty} \frac{\dd{\s}}{\s} \frac{J_1(\s\abs{\rbs+\ubs})}{\abs{\rbs+\ubs}}\frac{J_1(\s\abs{\rbs-\ubs})}{\abs{\rbs-\ubs}}
\end{equation}
We can evaluate the $\sigma$ and $\xi$ integrals using Eqs.~\ref{eq:Integral2} and ~\ref{eq:Integral3}:
\begin{equation} 
\begin{split}
    \frac{\Imc}{(2\pi)^3 e^4}&=\thetaMeasure\int_{-\infty}^{\infty}\int_{-\infty}^{\infty}\dd{x}\dd{y} \dd{\rbs} L(\sqrt{W}(x\ubs+y\vbs),\sqrt{W}\ubs) \frac{\log\abs{\frac{1}{\rbs\cdot\ubs}}}{2\max(\abs{\rbs+\ubs}^2,\abs{\rbs-\ubs}^2)}
\end{split}
\end{equation}
where 
\begin{equation}
    L(\abss,\bbs)= 
    \frac{1}{\abs{\bbs}} \arcsinh\left(\frac{\bbs\cdot(\abss +  \bbs)}{\abs{\bbs}\abs{\abss\times\bbs}}\right)
    -
    \frac{1}{\abs{\bbs}} \arcsinh\left(\frac{\bbs\cdot(\abss- \bbs)}{\abs{\bbs}\abs{\abss\times\bbs}}\right).
\end{equation}
To proceed, we take the following coordinates $\rbs = x \ubs  + y \vbs$  where $\vbs = [-\sin(\theta),\cos(\theta)]$ so that $\vbs \cdot \ubs=0$. {One can check that the Jacobian is trivial by writing $\rbs_{\text{old}} = R_\theta \rbs_{\text{new}}$ and then $\dd^2{\rbs_{\text{old}}}\wedge \dd{\theta} = \det(R_{\theta})\dd^2{\rbs_{\text{new}}}\wedge \dd{\theta} $ and noting that $\det(R_\theta)=1$}. Then
\begin{equation} 
    \frac{\Imc}{(2\pi)^3 e^4}= \thetaMeasure \int_{-\infty}^{\infty}\int_{-\infty}^{\infty}\dd{x}\dd{y} \dd{\rbs} L(\sqrt{W}(x\ubs+y\vbs),\sqrt{W}\ubs) \frac{\log\abs{\frac{1}{\rbs\cdot\ubs}}}{2\max(\abs{\rbs+\ubs}^2,\abs{\rbs-\ubs}^2)}
\end{equation}
We split the $\rbs$ integral in four quadrants by introducing signs $s_{x/y}=\pm 1$ and integrate over $x>0,y>0$.
\begin{equation}
    \begin{split}
        \frac{\Imc}{(2\pi)^3 e^4}&= \sum_{s_x,s_y}\thetaMeasure\int_{0}^{\infty}\int_{0}^{\infty}\dd{x}\dd{y} L(\sqrt{W}(x s_x \ubs+y s_y\vbs),\sqrt{W}\ubs) \frac{-\log(x)}{2[(x+1)^2+y^2]}
    \end{split}
\end{equation}
where we used that $\max(\abs{\rbs+\ubs}^2,\abs{\rbs-\ubs}^2)= y^2+(x+1)^2$. {To show this note that $\max((x+ 1)^2, (x-1)^2)= \max(x^2 +2x +1, x^2-2x +1)= x^2+1 + 2\max(x,-x) =x^2+2\abs{x}+1 = (\abs{x}+1)^2 $ and we can then replace $x$ by $x s_x$. }

We next want to massage the expressions to be able to perform some of the integrals over infinite regions. We start by noting that 
\begin{equation}
    L(\sqrt{W}(xs_x\ubs + ys_y\vbs),\sqrt{W}\ubs)= \frac{1}{\abs{\ubs}_{W}}\sum_{k=\pm 1}  k \arcsinh\left(\frac{\ubs\cdot W \cdot ( (s_x x+k)\ubs + s_y y \vbs )}{\abs{\ubs}_W \abs{y}}\right)
\end{equation}
where we have used that $\abs{(A \abss)\times (A\bbs)} = \abs{\abss \times \bbs}$ when $\det(A)=1$. We make a further change of variables $y\to \rho(1+x)$:
\begin{equation}
    L(\sqrt{W}(xs_x\ubs + ys_y\vbs),\sqrt{W}\ubs)= \frac{1}{\abs{\ubs}_{W}}\sum_{k=\pm 1}  k\arcsinh\left(\frac{\ubs\cdot W \cdot (\frac{ (s_x x+k)}{ (1+x)}\ubs + s_y \rho \vbs )}{\abs{\ubs}_W \rho}\right)
\end{equation}
followed by $\zeta = \frac{x-1}{x+1}$ so that $\frac{s_x x +k}{1+x} = s_x \zeta^{\frac{1-s_x k}{2}}$. Then 
\begin{equation}
    L(\sqrt{W}(xs_x\ubs + ys_y\vbs),\sqrt{W}\ubs)
    = \frac{1}{\abs{\ubs}_{W}}\sum_{l=0,1}  (-1)^l\arcsinh\left(\frac{\ubs\cdot W \cdot (\z^l\ubs + s_xs_y \rho \vbs )}{\abs{\ubs}_W \rho}\right)
\end{equation}
where we changed variables in the sum from $k$ to $l$ defined by the relation $ks_x = (-1)^l$. The last expression only depends on $s_xs_y$ so we redefine $s=s_xs_y$ and replace the sums as $\sum_{s_x,s_y} \to 2\sum_{s}$.

Then 
\begin{equation}\label{eq:DefiningIFG}
    \begin{split}
        \frac{\Imc}{(2\pi^3)e^4} &= -2\thetaMeasure  \int_{-1}^{+1} \frac{\dd{\zeta}}{1-\zeta} \log\left(\frac{1+\z}{1-\z}\right) F(\theta, \zeta)\\
        F(\theta, \zeta) &= \int^{\infty}_{0}  \left[G(\theta,1,\rho)-G(\theta,\zeta,\rho)\right]\frac{\dd{\rho}}{1+\rho^2} \\
        G(\theta,\zeta,\rho) &= \frac{1}{2}\sum_{s=\pm 1}  \frac{1}{\abs{\ubs}_{W}}\arcsinh\left(\frac{\ubs\cdot W \cdot (\z\ubs + s \rho \vbs )}{\abs{\ubs}_W \rho}\right)
    \end{split}
\end{equation}
where the $\theta$ dependence of $G$ is via $\ubs$ and $\vbs$. $G$ is an odd function of $\zeta$ which implies  
\begin{equation}\label{eq:RelationForF}
    F(\theta,\zeta)+F(\theta,-\zeta) = 2F(\theta,0).
\end{equation}

Note that 
\begin{equation}
    \pdv{G(\theta,\zeta,\rho)}{\zeta} =\frac{1}{2}\sum_{s=\pm 1} \frac{1}{\sqrt{\r^2+ ( \abs{\ubs}_W \zeta+ \b s\r )^2}}; \beta = \frac{\ubs \cdot W\cdot \vbs}{\abs{\ubs}_W}.
\end{equation}
We can then use the $s$ variable to extend the integral of $\r$ to negative numbers and identify $\abss = \zeta\abs{\ubs}_W [1,0]$ and $\bbs = \beta [1,0]$ in the definite integral in Eq.~\ref{eq:Integral5} to obtain
\begin{equation}\label{eq:Exchange2_AngularPart}
f(\theta,\zeta)\equiv\int_{0}^{\infty}\pdv{G(\theta,\zeta,\rho)}{\zeta} \frac{\dd{\rho}}{1+\r^2}= \Re\left( \frac{\arctanh\left(\frac{\sqrt{1+(\b + \ii \zeta\abs{\ubs}_W)^2}\sqrt{1+\b^2}}{1+\b(\b + \ii \zeta\abs{\ubs}_W) }\right)}{\sqrt{1+(\b + \ii \zeta\abs{\ubs}_W)^2}}\right) ; \beta = \frac{\ubs \cdot W \cdot \vbs}{\abs{\ubs}_W}
\end{equation}
From the integral definition of $f(\theta,\zeta)$, one can see that $f(\theta,\zeta)$ is finite for $\zeta\in (0,1]$ and $f(\theta,\zeta \to 0)$ diverges logarithmic in $\zeta$. Therefore, assuming there is no issue in swapping the $\rho$ and $\zeta$ integrations, we find that for $\zeta \in [0,1]$
\begin{equation}
    F(\theta,\zeta) = \int_{\zeta}^1 f(\theta,\zeta')\dd{\zeta'}.
\end{equation}
Using the properties of $f$, we deduce that $F$ is finite for all values of $\zeta \in [0,1]$ and $F(\theta,\zeta) \propto (1-\zeta)$ for $\zeta $ close to $1$. 

Using the relation $F(\theta,\zeta)+ F(\theta,-\zeta) = 2F(\theta,0)$ in Eq.~\ref{eq:RelationForF}, we get:
\begin{equation}
    \begin{split}
     \frac{\Imc}{-2(2\pi^3)e^4} &= \thetaMeasure\left[ \int_{0}^1 \frac{2\dd{\zeta}}{1-\z^2} \log\frac{1-\z}{1+\z} F(\theta,\zeta) + 2 F(\theta,0) c\right]\\
     c &= \int_{0}^1 \frac{\dd{\zeta}}{1+\z}\log \frac{1-\z}{1+\z} = - \frac{\pi^2}{12}
    \end{split}
\end{equation}
Next, we use that $\frac{\dd}{\dd \zeta} \log\frac{1-\z}{1+\z}  = \frac{2}{1-\z^2}$ to integrate by parts 
\begin{equation}
        \frac{\Imc}{-2(2\pi^3)e^4} =\thetaMeasure  \left[ -\frac{1}{2}\int_{0}^1 \dd{\zeta} \left(\log\frac{1-\z}{1+\z}\right)^2 \pdv{F(\theta,\zeta)}{\z} -\frac{\pi^2}{6 } F(\theta,0) \right]
\end{equation}
The function $\left(\log\frac{1-\z}{1+\z}\right)^2 {F(\theta,\zeta)} $ vanishes at $\zeta=0,1$ so only the integral term of the integration by parts survive. Next using the integral representation of $F$ in terms of $f$, allows us to write exchange integral as the integration over an angle and variable in the interval $(0,1)$: 
\begin{equation}\label{eq:FinalImc}
        \frac{\Imc}{-2(2\pi^3)e^4} =\thetaMeasure  \left[ \frac{1}{2}\int_{0}^1 \dd{\zeta} \left(\log\frac{1-\z}{1+\z}\right)^2 f(\theta,\z) -\frac{\pi^2}{6 } \int_{0}^1 f(\theta,\zeta) \dd{\zeta} \right]
\end{equation}

For the isotropic system ($W = \Id$), $\beta=0$ and $\abs{\ubs}_W =1$. So that the function $f$ reduces to $\Re( \arctanh(\sqrt{1-\z^2})/\sqrt{1-\z^2})  $ which equals $\log(\frac{1+\sqrt{1-\z^2}}{1-\sqrt{1-\z^2}})/2\sqrt{1-\z^2}$ for $\z\in (-1,1)$. This is consistent with the expression for $F'(z)$ in Appendix B of Ref.~\cite{Ishihara1980ExactExchange2ndOrder}.

Lastly, by comparing Eq.~\ref{eq:FinalImc} to Equation 22 in Ref.~\cite{Ishihara1980ExactExchange2ndOrder} and their value of exchange energy, we find that 
\begin{equation}\label{eq:ExchangeEnergy2PT}
    \begin{split}
       \varepsilon_{0,\text{x}} &= \Xsc(\eta) \equiv\frac{1}{6}\int_{0}^{1} \bar{f}(\zeta)\times \left(1-\frac{3}{\pi^2}\left(\log \frac{1-\z}{1+\z}\right)^2\right)\dd{\zeta}; \\
        \bar{f}(\zeta)&=\left(\int_{0}^{2\pi } \frac{\dd{\theta}}{2\pi} \frac{f(\zeta,\theta)}{\sqrt{\cos(\theta)^2 /\sqrt{\eta} + \sin(\theta)^2\sqrt{\eta}}}  \right).
    \end{split} 
\end{equation}
where $f(\zeta,\theta)$ is defined in Eq.~\ref{eq:Exchange2_AngularPart}.

We numerically evaluate $\Xsc(\eta)$ for $1/32<\eta \leq 1$ and plot the result in Fig.~\ref{fig:Exchg2Func}. We found $\Xsc(1) \approx 0.1143$ which is consistent to the analytic value $\beta(2) -\frac{8}{\pi^2}\beta(4)$ cited in Ref.~\cite{loos2016uniform}. ( $\beta(s)$ is the Dirichlet $\beta$ function.) 
\begin{figure}[h]
    \centering
    \includegraphics[width=0.45\textwidth]{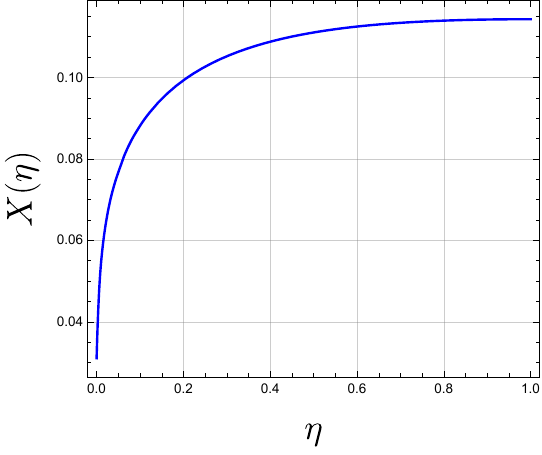}
    \caption{ Plot of $\Xsc(\eta)$ as defined in Eq.~\ref{eq:ExchangeEnergy2PT}}
    \label{fig:Exchg2Func}
\end{figure}

\subsection{Anomalous contribution}

According to Kohn and Luttinger\cite{KohnLuttinger1960}, the computation of the many-body ground state energy at zero temperature of a non-spherically symmetric fermionic system is subtle. In particular, the energy obtained by taking the $T \to 0 $ limit before the thermodynamic limit ($V\to \infty$ with $\nel = \frac{N}{V}$) will differ from the energy obtained taking the thermodynamic limit first and then letting $T\to 0$. This is not an issue for spherically symmetric systems as shown by Kohn and Luttinger in Ref.~\cite{LuttingerWard1960}.

The difference between both ways of calculating the energy can be traced back to the renormalization of the chemical potential. To 1PT, there is no difference. The leading difference between the two methods is (Eq.~30 of \cite{KohnLuttinger1960}):
\begin{equation}
    E_{2}' = \frac{1}{2}\left(\sum_{s}\delta(\eps_s -\mu_{0}) \right)\times \left(\frac{\left[\sum_{s}\delta(\eps_s-\mu_0) \Qsc_s\right]^2}{\left[\sum_{s}\delta(\eps_s-\mu_0)\right]^2}
    -
    \frac{\sum_{s}\delta(\eps_s-\mu_0) \Qsc_s^2}{\sum_{s}\delta(\eps_s-\mu_0)}\right).
\end{equation}
The sums are over the single-particle states labeled by $s$ and $\mu_0$ is the chemical potential. The state $s$ has a flavor $\alpha_s$ index and momentum $\qbs_{s}$. $\Qsc_s$ is given by
\begin{equation}
    \begin{split}
        \Qsc_{(\a,\qbs)}&= \int_{\eps_{\a}(\kbs) < \mu_0} \frac{\dd^2{\kbs}}{(2\pi)^2} \tilde{v}(\kbs-\qbs).
    \end{split}
\end{equation}
Let $\mu_0 = \frac{\kF^2}{2m}$, $q_x = \h_{\a}^{1/4}\kF \bar{q}\cos(\varphi)$ and $q_y = \h_{\a}^{-1/4}\kF \bar{q}\sin(\varphi)$, then $\Qsc_{\a,\qbs} = \frac{\kF}{m\aB}\vmf(\bar{q},\varphi;\h_\a)$ with 
\begin{equation}
    \vmf(\bar{q},\varphi; \eta) = \int_{0}^{2\p}\frac{\dd{\q}}{2\p}\frac{\Re[\sqrt{1- \bar{q}^2 \cos(\q-\varphi)^2}]}{\sqrt{\cos(\q)^2/\sqrt{\h}+ \sin(\q)^2\sqrt{\h}}}.
\end{equation}
The above expression is obtained by representing $\tilde{v}(\kbs-\qbs) = \int \dd{\rbs}\ee^{\ii\kbs\cdot\rbs} \ee^{-\ii\qbs\cdot\rbs} \frac{1}{m\aB \abs{\rbs}}$ and doing the $\kbs$ integral in the expression for $\Qsc_{\a,\qbs}$. It turns out that we can evaluate $\smf(\varphi;\h)\equiv \vmf(\bar{q}=1,\varphi;\h)$ analytically:
\begin{equation}
    \begin{split}
        \smf(\varphi;\h) &= \int_{\varphi}^{\varphi+\p}\frac{\dd{\q}}{\p} \frac{\sin(\q-\varphi)}{\sqrt{\cos(\q)^2/\sqrt{\h}+ \sin(\q)^2\sqrt{\h}}}\\
        &= \frac{2}{\p} \left( \cos(\varphi)^2\frac{h(\sqrt{1-1/\h}\cos(\varphi))}{\h^{+1/4}} + \sin(\varphi)^2\frac{h(\sqrt{1-\h}\sin(\varphi))}{\h^{-1/4}}\right);\\
        h(z) &\equiv \frac{\arcsin(z)}{z}; \quad \abs{\Re[z]} \leq 1.
    \end{split}
\end{equation}
We show $\smf(\varphi;\eta)$ in Fig.~\ref{fig:smf:app}.

\begin{figure}
    \includegraphics[width=0.9\textwidth]{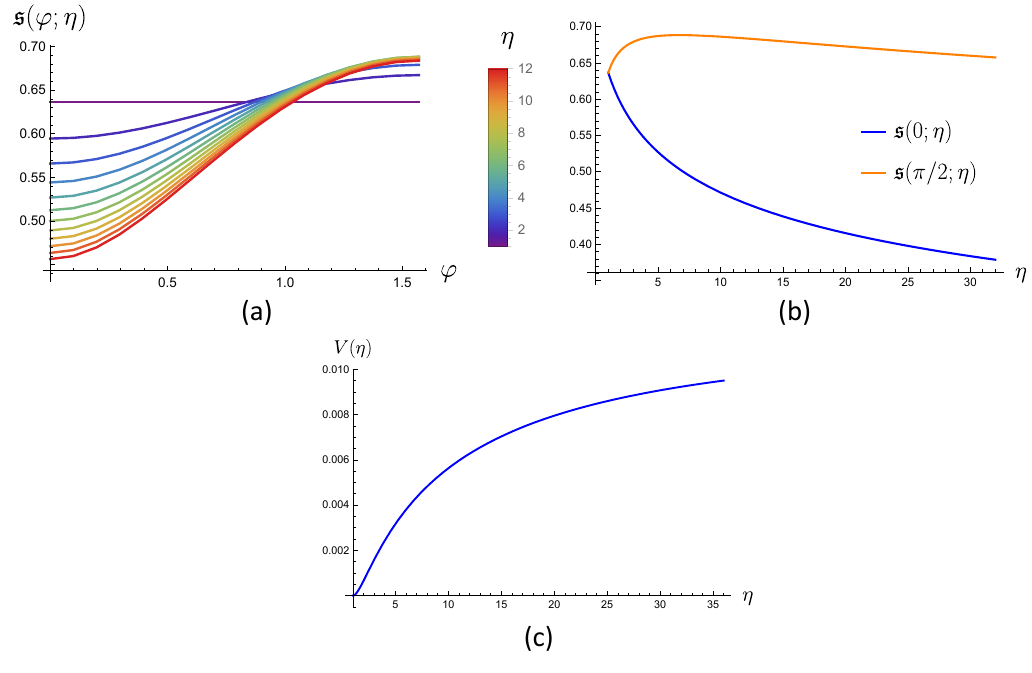}
    \caption{\textbf{(a)} Shows the angular dependence of $\smf$ for several values of $\eta$. \textbf{(b)} Shows the expected minima and maxima of $\smf$ for a range $\eta$'s. \textbf{(c)} $\Vsc(\eta$ is the variance of $\smf(\varphi;\eta)$ as a function of the angle $\varphi$.} 
    \label{fig:smf:app}    
\end{figure}

We now move back to the calculation of $E'$. In the thermodynamic limit, we can replace the sums over states by sums over flavors and integrals over momentum:
\begin{equation}
    \sum_s \delta(\eps_s - \mu_0)(\dots)  \to  \A \sum_\a \int\frac{\dd^2\qbs}{(2\p)^2}\delta(\eps_{\qbs}-\mu_0)(\dots) \to \frac{\A m}{2\pi}\sum_\a \int_0^{2\pi}\frac{\dd\varphi}{2\pi}(\dots)\eval_{\qbs =\sqrt{2m\mu_0}[\h_{\a}^{1/4}\cos(\varphi),\h^{-1/4}_{\a}\sin(\varphi)]}
\end{equation}
where in the last step, we used the following change of variables  $\qbs =q[\h_{\a}^{1/4}\cos(\varphi),\h_{\a}^{-1/4}\sin(\varphi)]$ and evaluated the $q$ integral so the $(\dots)$ are evaluated at $q = \sqrt{2m \mu_0}$. 

For the `diagonal' states, we obtain 
\begin{equation}\label{eq:E2xp}
    \frac{E_{2}'}{\A} = \frac{\nel}{m\aB^2} \expval{ [\smf(\varphi_1;\h)-\expval{\smf(\varphi_2;\h)}_{\varphi_2}]^2 }_{\varphi_1},
\end{equation}
where $\expval{f(\varphi)}_{\varphi}\equiv \int_0^{2\pi}\frac{\dd{\varphi}}{2\pi} f(\varphi)$. We have used that $\expval{\smf(\varphi;\eta)}_{\varphi}=\expval{\smf(\varphi;1/\eta)}_{\varphi}$. We thus see that $E_2'$ is the same for all the diagonal states. In Fig.~\ref{fig:smf:app} (c), we show the $\eta$ dependence of the variance
\begin{equation}\label{eq:Vsc:def}
    \Vsc(\eta)\equiv 
    \int_{0}^{2\p}\!\frac{\dd{\varphi_1}}{2\pi}
    \int_{0}^{2\p}\!\frac{\dd{\varphi_2}}{2\pi} \left[\smf(\varphi_1;\eta)^2 - \smf(\varphi_1;\eta)\smf(\varphi_2;\eta)\right].
\end{equation}
$\Vsc(\eta)$ seems to be and increasing function, but it should eventually decrease because $0\leq \smf(\varphi;\eta)\leq \Ksc(\eta)$ and $\Ksc(\eta)$ tends to zero in the large anisotropy limit $\abs{\log(\eta)}\to \infty$.

Dividing Eq.~\ref{eq:E2xp} by $\tfrac{\nel}{(m\aB^2)}$, we obtain that the anomalous exchange contribution to the energy per particle, $E_{\text{x}'}$ is 
\begin{equation}
    E_{\text{x}'} = \varepsilon_{x'}= \Vsc(\h).  
\end{equation}

\subsection{Random phase approximation for the energy}

We move on to the evaluation to the correlation energy within the RPA. The correlation energy per electron is 
    \begin{align}\label{eq:RPAenergy:app}
        \frac{E_{\rr}}{\Nel}
        &=  \frac{1}{\nel}\int_{0}^{\infty}\!\!\frac{\dd{\Omega}}{2\pi} \int\!\!\!\frac{\dd^2{\qbs}}{(2\pi)^2} \log(1+\tilde{v}(\qbs) \Pi(\ii\Omega,\qbs)) - \tilde{v}(\qbs) \Pi(\ii\Omega,\qbs)\\
        & = - \frac{1}{\nel}\int_{0}^{\infty}\!\!\frac{\dd{\Omega}}{2\pi} \int\!\!\!\frac{\dd^2{\qbs}}{(2\pi)^2} \int_{0}^1 \!\!\dd{\l}\frac{\l \left[\tilde{v}(\qbs) \Pi(\ii\Omega,\qbs)\right]^2}{1+ \l  \tilde{v}(\qbs) \Pi(\ii\Omega,\qbs)}
    \end{align}
In the $(x,y,\q)$ variables, we get instead 
\begin{align}\label{eq:cRPA:xyq}
\frac{E_{\rr}}{\Nel} &=
\frac{2}{\pi}
\int_{0}^{\infty}\!\!\!\dd{x}
\int_{0}^{\infty}\!\!\!\dd{y}
\int_{0}^{2\pi}\!\!\frac{\dd{\q}}{2\pi} \left(\frac{4x}{\rs}\right)^2\left[
\log(1 + \frac{\rs}{4x}\Psi(x,y,\q) ) - \frac{\rs}{4x}\Psi(x,y,\q) \right] \\
&= -\frac{1}{\pi}
\int_{0}^{\infty}\!\!\!\dd{x}
\int_{0}^{\infty}\!\!\!\dd{y}
\int_{0}^{2\pi}\!\!\frac{\dd{\q}}{2\pi} \int_{0}^1 \!\!\dd{\l}
\frac{2\l \left[\Psi(x,y,\q)\right]^2}{1+ \l \frac{\rs}{4x}\Psi(x,y,\q) } 
      \end{align}
where $\Psi = \sum_{\a=1}^{4}\Psi_\a$. Using the fact that
\begin{equation}
    \frac{\left[\Psi(x,y,\q)\right]^2}{1+ \l \frac{\rs}{4x}\Psi(x,y,\q) }  \leq \left[\Psi(x,y,\q)\right]^2; \quad (\l\geq 0),
\end{equation}
together with Eq.~\ref{eq:cRPA:xyq} and Eq.~\ref{eq:TwoLoopIntegral:Ver2}, we find that
\begin{equation}
     \varepsilon_{0,r}\leq  E_{\rr} \leq  0.
\end{equation}

We can use the change of variable $(x,y)\to (x/\sqrt{2},y/\sqrt{2})$ to show that 
\begin{equation}
    \begin{split}
        E_{\rr}^{\text{Sym}}(\rs) &= 2E_{\rr}^{\text{SP}}(2\sqrt{2}\rs) \\
        E_{\rr}^{\text{VP}}(\rs) &= 2E_{\rr}^{\text{SVP}}(2\sqrt{2}\rs)
    \end{split}
\end{equation}
The overall factor of two is $(1/2)*(2)^2$, where $1/2$ comes from the Jacobian, and the $(2)^2$ comes from $\Psi^2$. Similarly, the factor of $\sqrt{2}$ in the argument comes from $x$ and the factor of $2$ comes from $\Psi$. 

The RPA correlation energy for VP can be written as an average of the RPA energy of the isotropy 2DEG ($\eta=1$) by noting that we can make the change of variables to $(x,y)\to (x /I , y I)$ with $I = I (\eta,\q)$. Then 
\begin{equation}\label{eq:RPA_VP:2DEG}
    E_{\rr}^{\text{VP}}(\eta,\rs) = \int_0^{2\pi}\frac{\dd{\q}}{2\pi}E_{\rr}^{\text{2DEG}}(\rs I(\h,\q))
\end{equation}
Using this formulation, it is easier to evaluate the integral for large $\eta$, once we know the large $\rs$ behaviour of the RPA integral \cite{ioriatti1981ground}. We present the results in Fig.~\ref{fig:ScalingLargeAnisotropyVP}.

\begin{figure*}[t]
    \centering
    \includegraphics[width=0.7\textwidth]{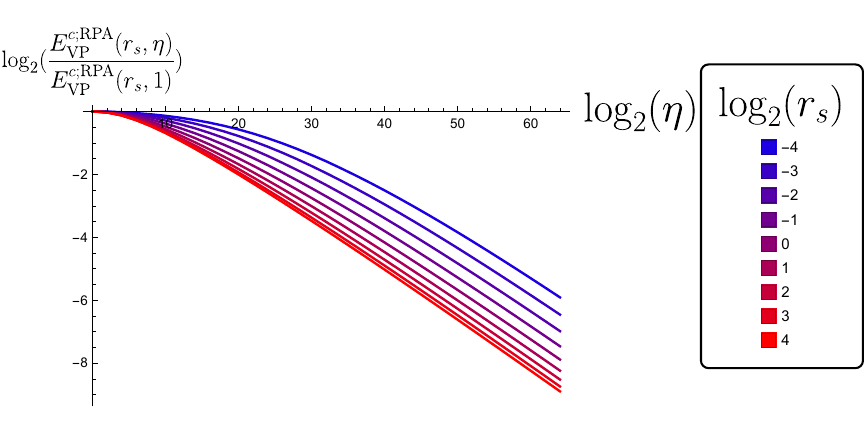}
    \caption{Scaling behaviour for the RPA correlation energy in the VP sector evaluated using Eq.~\ref{eq:RPA_VP:2DEG}. }
    \label{fig:ScalingLargeAnisotropyVP}
\end{figure*}

\subsubsection{Evaluation of \texorpdfstring{$\varepsilon_{1,\ell}$}{vareps1l}}\label{app:LogarithmicRPA}

In order to evaluate $\varepsilon_{1,\ell}$, we fix $\eta$ and study the small $\rs$ behaviour of $E_{\rr}$. Following Ref.~\cite{Giuliani2007ScalingLogCorrection}, the $\rs \log(\rs)$ term comes from the small $x$ region of $E^{\rr} - \varepsilon_{0,r}$. We start from the expression
\begin{equation}
    E_{\rr} - \varepsilon_{0,\rr}=  \frac{\rs}{4\pi}
    \int_{0}^{\infty}\!\!\!\dd{x}
    \int_{0}^{\infty}\!\!\!\dd{y}
    \int_{0}^{2\pi}\!\!\frac{\dd{\q}}{2\pi} \int_{0}^1 \!\!\dd{u}
\frac{2u^2 \left[\Psi(x,y,\q)\right]^3}{x+ u \frac{\rs}{4}\Psi(x,y,\q) }.
\end{equation}
We then set $x=0$ in the argument of $\Psi(x,y,\theta)$ and use that
\begin{equation}
    \int_0^{x_0} \dd{x}  \frac{1}{x+ u \frac{\rs}{4} \Psi} = -\log(\rs)+\dots
\end{equation}
where $x_0$ is a cutoff and $\dots$ refer to subleading terms in $\rs$. Doing the $u$ integral and dividing by $\rs\log(\rs)$, we obtain 
\begin{equation}\label{eq:LogCoulomb}
    \varepsilon_{1,\ell} = -\frac{1}{6\pi}\int_{0}^{\infty}\!\!\!\dd{y}
    \int_{0}^{2\pi}\!\!\frac{\dd{\q}}{2\pi} 
    \left[\Psi(0,y,\theta)\right]^3.
\end{equation}
We have $\Psi(0,y,\theta) = \sum_\a R(\frac{y}{\sqrt{\nu_{\a} I(\eta_\a,\theta)}})$ where $R(y) = 1-\frac{1}{\sqrt{1+1/y^2}}$.

Therefore, Eq.~\ref{eq:LogCoulomb} can be written in terms of the function \footnote{Mathematica gives an expression for the indefinite integral. So in principle, one could give an explicit form of the function. } 
\begin{equation}
    \Gmc(s_1,s_2,s_3) = \int_{-\infty}^{+\infty} R(u/s_1)R(u/s_2)R(u/s_3) \dd{u}.
\end{equation}
which is symmetric under the permutation of its arguments. When two arguments of $\Gmc$ are equal, the function can be simplified to \cite{Giuliani2007ScalingLogCorrection}
\begin{equation}\label{eq:DefinitionVarPhi}
    \begin{split}
        \Fmc(s_1,s_2) &\equiv \Gmc(s_1,s_1,s_2) = s_1 \varphi(s_2/s_1); \\
        \varphi(x) &= 4(1+x)-\pi - 4\Eell(1-x^2) + 2 \frac{\arcsin(\sqrt{1-x^2})}{\sqrt{1-x^2}}.
    \end{split}
\end{equation}
with $\varphi(1) = 10-3\pi$. Then, we can write 
\begin{equation}\label{eq:varepsilon:0l}
    \begin{split}
        \varepsilon_{1,\ell} 
        &= \sum_{\a_1,\a_2,\a_3}\varepsilon_{0,\ell}(\a_1,\a_2,\a_3); \\
        \varepsilon_{1,\ell}(\a_1,\a_2,\a_3)&=
        -\frac{1}{12\pi} \int_{0}^{2\pi}\!\!\frac{\dd{\q}}{2\pi} \Gmc(k_{\a_1}(\q), k_{\a_2}(\q), k_{\a_3}(\q)) ;\\
        k_{\a}(\q)&= \sqrt{\nu_\a}I(\eta_\a,\q).
    \end{split}
\end{equation}
As we only consider fully polarized states $\Gmc$ at most two of the arguments of $\Gmc$ are different. We can write the results in terms of 
\begin{equation}\label{eq:DefEandF}
    \begin{split}
        \mathscr{E}(\eta) &= \int_{0}^{2\pi}\frac{\dd{\q}}{2\pi} \sqrt{\sqrt\h \cos(\q)^2+\sin(\q)^2/\sqrt\h}= \frac{2}{\pi \h^{1/4}}\Eell(1-\h)\\
        \mathscr{F}(\eta) &= \int_{0}^{2\pi}\frac{\dd{\q}}{2\pi} \sqrt{\sqrt\h \cos(\q)^2+\sin(\q)^2/\sqrt\h} \frac{\varphi\left(\sqrt{\frac{\h \sin(\q)^2+\cos(\q)^2}{\h \cos(\q)^2+\sin(\q)^2}}\right)}{\varphi(1)}
    \end{split}
\end{equation}
Numerical evaluation shows that $\Esc(\eta)>\Fsc(\eta)$ for $\eta\in(0,1)$. Using that $\Gmc(as_1,as_2,as_3) = a\Gmc(s_1,s_2,s_3)$ for $a>0$, we obtain when all the anisotropies are different ($\eta_{\a_1}=\eta_{\a_2}=\eta_{\a_3}$),
\begin{equation}
    \varepsilon_{1,\ell}(\a_1,\a_2,\a_3) = - \frac{10-3\pi}{12\pi}\sqrt{\nu_\a} \Esc(\eta).
\end{equation}
When the anistropies are different, we obtain instead
\begin{equation}
    \varepsilon_{1,\ell}(\a_1,\a_2,\a_3) = - \frac{10-3\pi}{12\pi}\sqrt{\nu_\a} \Fsc(\eta).
\end{equation}
Combining the above two results with Eq.~\ref{eq:varepsilon:0l}, we obtained the results nin Tab.~\ref{tab:summaryExpansionParams}. We plot the functions $\mathscr{E}$, $\mathscr{F}$ and $\mathscr{F}/\mathscr{E}$ in Fig.~\ref{fig:Log}.

\begin{figure}[h]
    \centering
    \includegraphics[width=0.8\textwidth]{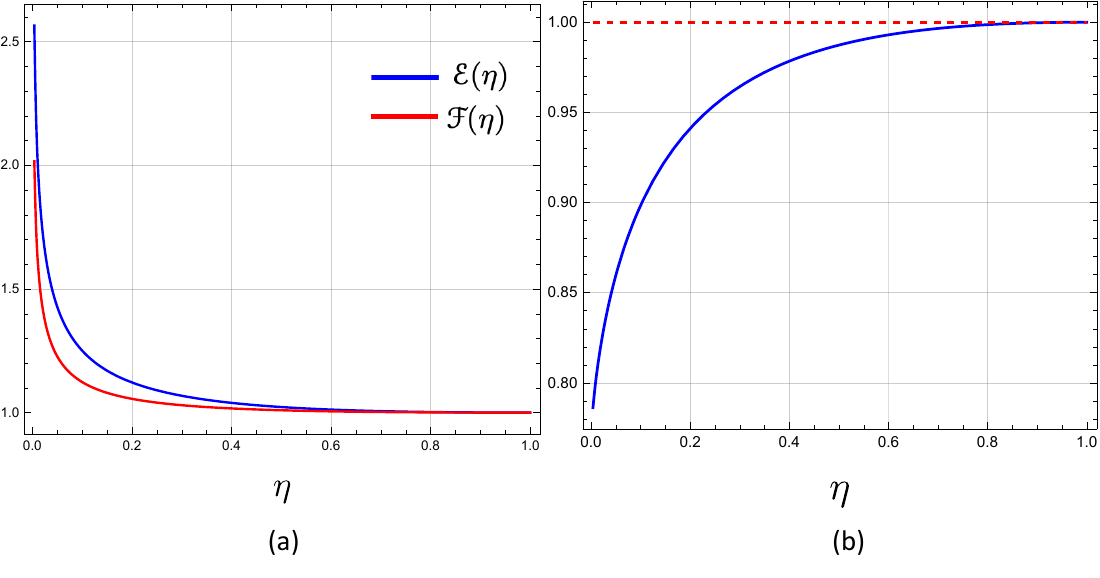}
    \caption{ \textbf{(a)} Plot of $\Esc(\eta)$ and $\Fsc(\eta)$ as defined in Eq.~\ref{eq:DefEandF}. \textbf{(b)} Plot of the ratio. }
    \label{fig:Log}
\end{figure}

\subsection{Polarization dependence of the energy and estimation of valley susceptibility}\label{app:PolarizationDependence}

The ring contribution to the energy for the partially valley-polarized state (pVP) at polarization $\zeta = (n_{X}-n_Y)/\nel$ can be obtained by evaluating (\ref{eq:RPAenergy:app}) using 
\begin{equation}
    \Pi(\ii\Omega, \qbs) =  \frac{m}{2\pi} \left(4 -2 \frac{\Re\sqrt{z_X^2-\tfrac{1+\z}{4}}}{\Re z_X} -2 \frac{\Re\sqrt{z_Y^2-\tfrac{1-\z}{4}}}{\Re z_Y}\right),
\end{equation}
where we used the same notation as App.~\ref{app:Varepsilon:0r}. We denote this contribution by $E^{\text{pVP}}_{\rr}[\zeta;\rs,\h]$. 

To confirm the first-order nature of the transition between the symmetric state to the valley-polarized state in Fig.~\ref{fig:Fig1}, we evaluated 
\begin{equation}
    \frac{E^{\text{pVP}}_{\RPA}[\zeta;\rs,\eta]}{\A} = \frac{1+\z^2}{4\rs^2} - \frac{4\sqrt{2}\Ksc(\h)}{3\p \rs}\left(\left(\frac{1+\z}{2}\right)^{3/2} + \left(\frac{1-\z}{2}\right)^{3/2}\right)  + \frac{E^{\text{pVP}}_{\rr}[\zeta;\rs,\h]}{\A} \quad [\mathrm{Ha}_*]
\end{equation}
where the first term is the kinetic energy, the second term is the Fock energy, and the last term is ring contribution to the energy. 

We evaluated $\frac{E^{\text{pVP}}_{\RPA}[\zeta;\rs,\eta]}{\A}$ for $\zeta^2 = j /8$, with $j=0,1,\dots,8$, and some values of $\eta$ and $\rs$ near the symmetric to valley-polarized transition line in the RPA phase diagram (Fig.~\ref{fig:Fig1}). We found that the state with the lowest energy has $\zeta^2=0$ or $\zeta^2=1$, {i.e.} it is either the symmetric state of the valley-polarized state. We show the $\zeta^2$ dependence of the energy for $\eta=3$ and some values of $\rs$ in Fig.~\ref{fig:polarizationNsusceptibility}.

Additionally, we extracted the valley susceptibility through a linear fit of $\frac{E^{\text{pVP}}_{\RPA}[\zeta;\rs,\eta]}{\A}$ for $0 \leq \zeta^2 \leq 3/8$. We show the inverse susceptibility vs. $r_s$ for $\eta=3,4,5$ in Fig.~\ref{fig:polarizationNsusceptibility}. The valley susceptibility is enhanced upon approaching the transition, but remains finite at the transition point. 

\begin{figure}[h]
    \centering
    \includegraphics[width= 1\textwidth]{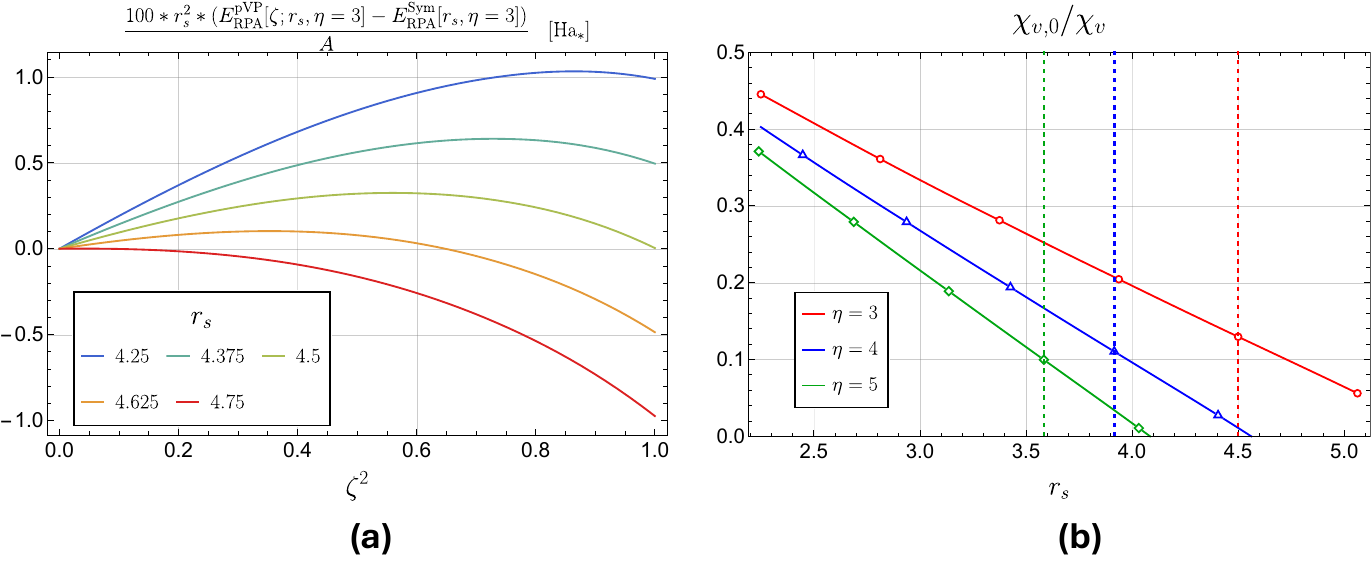}
    \caption{\textbf{(a)} Polarization dependence of energy (within the RPA) at $\eta=3$ near the transition from the symmetric state $(\zeta=0)$ to the valley-polarized state ($\zeta=1$). \textbf{(b)} Dependence of the valley-susceptibility $\chi_{v}$ as a function of polarization by fitting the polarization dependence of the energy. $\chi_{v,0}= \frac{2m}{\pi}$ is non-interacting valley susceptibility. Markers denote values obtained through the fit and lines serve as a guide to the eye. Vertical dashed lines denote the critical $\rs$ at the given anisotropy $\eta$. }
    \label{fig:polarizationNsusceptibility}
\end{figure}

\section{Finite temperature effects}\label{app:TemperatureEffects}

To calculate the finite temperature effect on the phase diagram in Fig.~\ref{fig:Fig1}, we ought to compare the Helmholtz free energy $F$ for the various states. Following previous calculations for the 3DEG \cite{Perrot1984XC3d} and 2DEG \cite{Phatisena1986XC2deg}, we approximate
\begin{equation}
    F(n,T) \approx F_{\RPA}(n,T) = F_0(n,T) + \Omega_{\xr}(\mu_0,T),
\end{equation}
where $n$ is the electron-density, $T$ is the temperature, $F_0$ is the non-interacting Helmholtz free energy, $\Omega_{\xr}(\mu_0,T)$ is the sum of the exchange and ring contributions to the Landau free energy evaluated at the non-interacting chemical potential $\mu_0$. 

Recall that the non-interacting Helmholtz free energy, to leading order in $T$, is 
\begin{equation}
    F_0(n,T)  = E_0 - \frac{mA}{2\p}\frac{\p^2 T^2}{6} \gmf_0.
\end{equation}
$\gmf_0$ is the number of occupied flavors, {e.g.} for the Sym state $\gmf_0=4$; and for VP and SP, $\gmf_0 =2$. In App.~\ref{app:EvaluationOmegaxr}, we show that, to leading order in $T$, the Landau thermodynamic potential is 
\begin{equation}
    \Omega_{\xr}(\mu_0,T) =  E_{\xr} - \frac{m \A}{2\p} \frac{\p^2T^2 }{6} \gmf_{\xr}.
\end{equation}
where $\gmf_{\xr}$ is a function of $\rs$, $\eta$ and state $a \in \StateList$.

Then, to leading order in $T$ the Helmholtz free energy within RPA is 
\begin{equation}
 \frac{F_{}}{\A} =  \frac{E_{}}{\A} - \frac{m}{2\p} \frac{\p^2 T^2}{6}\gmf_{};\quad \gmf_{} \equiv \gmf_{0} + \gmf_{\xr}.
\end{equation}

We interested in the location of the transition between the Sym and VP states. Therefore, we want to evaluate $\gmf_{\text{Sym}}-\gmf_{\text{VP}}$. We show in Fig.~\ref{fig:degeneracyChange} the results for some values of $\eta$ and see that there is not much $\eta$ dependence in the relevant range for our RPA phase diagram.  

\subsection{Leading temperature correction to the exchange-correlation Landau Free energy}\label{app:EvaluationOmegaxr}

In general, we can write the exchange-correlation potential as 
\begin{equation}
    \frac{\Omega_{\xc}[\mu,T]}{\A}= - \frac{1}{2}\int_0^1\dd{\l}\int \frac{\dd^2\qbs}{(2\p)^2} \tilde{v}(\qbs)\left[T\sum_{\Omega_n}\chi_{nn}(\qbs,\ii\Omega_n;\l,\mu,T)- \nel\right].
\end{equation}
By applying Poisson summation to convert the Matsubara sum into integrals and then deforming the integration contour of the oscillatory terms, we obtain
\begin{equation}
    \begin{split}
        \frac{\Omega_{\xc}[\mu,T]}{\A} &= \frac{E_{\xc}[\nel]}{\A}+\varpi_{\xc}^{\mathfrak{I}}[\mu,T] + \varpi_{\xc}^{\mathfrak{R}}[\mu,T];\\
        \varpi_{\xc}^{\mathfrak{I}}[\mu,T]&= -\frac{1}{2}\int_0^1\dd{\l} \int \frac{\dd^2{\qbs}}{(2\p)^2}\tilde{v}(\qbs) \int_{-\infty}^{\infty}\frac{\dd{\eps}}{2\p} \left[\chi_{nn}(\qbs,\ii\eps;\l,\m,T) - \chi_{nn}(\qbs,\ii\eps;\l,\m,T=0)\right] \\ 
        \varpi_{\xc}^{\mathfrak{R}}[\mu,T]&= -\int_0^1\dd{\l} \int \frac{\dd^2{\qbs}}{(2\p)^2} \tilde{v}(\qbs)\int_{0}^{\infty}\frac{\dd{\eps}}{2\p} \frac{\Im\left[\chi_{nn}(\qbs,+\eps+\ii 0^+;\l,\m,T) - \chi_{nn}(\qbs,-\eps+\ii 0^+;\l,\m,T)\right]   }{\ee^{\b\eps}-1}  \end{split}
\end{equation}
Here $\nel = \frac{m}{2\p \b} \sum_{\a} \log(\ee^{\b\m_\a}+1)$.

We now approximate the correlation contributions by the contributions from ring diagrams. We obtain 
\begin{equation}
    \begin{split}
\varpi_{\xc}^{\mathfrak{I}}[\mu,T]\approx \varpi_{\xr}^{\mathfrak{I}}[\mu,T]&= -\int \frac{\dd^2{\qbs}}{(2\p)^2} \int_{0}^{\infty}\frac{\dd{\eps}}{2\p} \log\left[\frac{1+ \tilde{v}(\qbs)\Pi(\ii\eps,\qbs; \m,T)}{1+\tilde{v}(\qbs)\Pi(\ii\eps,\qbs; \m,T=0)}\right];\\ 
\varpi_{\xc}^{\mathfrak{R}}[\mu,T]\approx\varpi_{\xr}^{\mathfrak{R}}[\mu,T]&= -\int \frac{\dd^2{\qbs}}{(2\p)^2} \int_{0}^{\infty}\frac{\dd{\eps}}{2\p} \frac{1}{\ee^{\b\eps}-1}\Dmc_{\xr}(\qbs,\eps; \m,T);\\
\Dmc_{\xr}(\qbs,\eps; \m,T)&= \int_0^1\dd{\l} \Im\left[ \frac{\tilde{v}(\qbs) \Pi(\eps+\ii 0^+,\qbs; \m,T)}{1+ \l \tilde{v}(\qbs)\Pi(\eps+\ii 0^+,\qbs; \m,T)} - (\eps \to -\eps)\right].
        \end{split}
\end{equation}
The contribution from the imaginary axis, $\varpi^{\mathfrak{I}}_{\xr}$, can be calculated based on the behavior of $\Dmc_{\xr}(\qbs,\eps)$ near $\eps=0$, including contributions from both the electron-hole continuum and the plasmon. The contribution from the real axis, $\varpi^{\mathfrak{R}}_{\xr}$, comes from the temperature dependence of the Lindhard function. 
\def\eh{\text{eh}}
\paragraph{Electron-hole continuum contribution:} This contribution to $\Dmc_{\xr}$ from the electron-hole continuum is 
\begin{equation}
    \Dmc_{\eh}(\qbs,\eps;\m) = \arctan\left(\frac{\tilde{v}(\qbs)\Pi''(\eps,\qbs;\m,T)}{1+\tilde{v}(\qbs)\Pi'(\eps,\qbs;\m,T)}\right) -\arctan\left(\frac{\tilde{v}(\qbs)\Pi''(-\eps,\qbs;\m,T)}{1+\tilde{v}(\qbs)\Pi'(-\eps,\qbs;\m,T)}\right).
\end{equation}
Expanding around $\eps=0$, we have
\begin{equation}
   \Dmc_{\eh}(\qbs,\eps;\m) = \Dmc_{\eh}^{(1)}(\qbs;\m)\eps + \Omc(\eps^3);\quad \Dmc_{\eh}^{(1)}(\qbs;\m)\equiv 2\left(\frac{\tilde{v}(\qbs)\partial_{\eps}\Pi''(\eps=0,\qbs;\m,T)}{1+\tilde{v}(\qbs)\Pi'(\qbs,\eps=0;\m,T)}\right).
\end{equation}

For a single isotropic flavor, we get $\partial_\eps\Pi''(q,0) = \frac{m}{2\p}\frac{1}{\sqrt{1- (q/2\kF)^2}} \frac{m}{\kF q} \Theta(1-q/2\kF) = \frac{m}{2\p}\frac{1}{4\mu} \frac{\Theta(1-x)}{x\sqrt{1-x^2}}$, where $x = q/2\kF$. In terms of $\rs = 2/(\kF \aB)$, the leading contribution to $\varpi_{\xr}^{\mathfrak{I}}$ is
\begin{equation}
    \begin{split}
        \varpi_{\xr:\eh}^{\mathfrak{I}}[\m,T]&= -\frac{\kF^2}{4\pi }\frac{4}{\m}\int_0^{1} \frac{\dd{x}}{\sqrt{1-x^2}} \frac{1}{1 + \tfrac{4x}{\rs}}\int_{0}^{\infty}\frac{\dd{\eps}}{2\p} \frac{\eps}{\ee^{\b\eps}-1} = -\frac{m}{2\p} \frac{\pi^2}{6\b^2} \Xmf_0(\rs)\\
        \Xmf_0(\rs)&\equiv\frac{2}{\p}\int_0^{\p/2} \frac{1}{1+\tfrac{4}{\rs}\cos(\q)}\dd{\q} = 
     \begin{cases}
        \frac{\rs}{2\p}\frac{\arctan(\sqrt{(\rs/4)^2-1})}{\sqrt{(\rs/4)^2-1}} \quad &,4<\rs;\\
        \frac{\rs}{2\p}\frac{\arctanh(\sqrt{1-(\rs/4)^2})}{\sqrt{1-(\rs/4)^2}} \quad  &,0<\rs<4.
    \end{cases}.
    \end{split}
\end{equation}

For the states ($a \in \{\text{Sym},\text{SP},\text{VP},\text{SVP}\}$) we care about, we write
\begin{equation}
    \varpi_{\xr:\eh}^{\mathfrak{R}}[\mu,T]_{a} = - \frac{m}{2\p } \frac{\p^2}{6\b^2
} \mathfrak{X}_a(\rs,\h).
\end{equation}
For $a=$ VP or SVP, we obtained
\begin{equation}
    \begin{split}
        \Xmf_{\text{VP}}(\rs,\h) &=  4\int_0^{2\p}\frac{\dd{\q}}{2\p}\int_0^{1/I(\h,\q)}\frac{x\dd{x}}{2\p} \frac{2}{\tfrac{2\sqrt{2}}{\rs} x+2} \frac{1}{xI(\h,\q)}\frac{1}{\sqrt{1-x^2 I(\h,\q)^2}} \\
        \Xmf_{\text{SVP}}(\rs,\h) &=  4\int_0^{2\p}\frac{\dd{\q}}{2\p}\int_0^{1/I(\h,\q)}\frac{x\dd{x}}{2\p} \frac{1}{\tfrac{4}{\rs}x+1} \frac{1}{xI(\h,\q)}\frac{1}{\sqrt{1-x^2 I(\h,\q)^2}}
    \end{split}
\end{equation}
where $I(\h,\q) = \sqrt{\sqrt{\h}\cos(\q)^2+\sin(\q)^2/\sqrt{\h}}$. One can perform a change of variables to show that 
\begin{equation}
   \begin{split}
        \Xmf_{\text{SVP}}(\rs,\h) &= \int_0^{2\p} \frac{\dd{\q}}{2\p}\Xmf_0\left(\frac{\rs}{I(\h,\q)}\right)  = \int_0^{2\p} \frac{\dd{\q}}{2\p}\Xmf_0\left(\rs I(\h,\q)\right) \frac{1}{I(\h,\q)^2};\\
        \Xmf_{\text{VP}}(\rs,\h)  &= \Xmf_{\text{SVP}}(\sqrt{8}\rs,\h). 
   \end{split}
\end{equation}
For the other states, we get a more complicated expression. We can similarly show that $\Xmf_{\text{Sym}}(\rs,\h) = \Xmf_{\text{SP}}(\sqrt{8}\rs,\h)$. Furthermore, $\Xmf_{\text{SP}}(\rs,\h)> \Xmf_{\text{VP}}(\rs,\h)$ because there are regions in the integration domain where the screening from the other valley becomes less effective (the term $\sqrt{1-(2\kF/q)^2}$ in the denominator turns on). Additionally, $\Xmf_{a}(\rs,\h)$ is an increasing function of $\rs$, which follows from $\partial_{\rs}\Xmf_a>0$.

\paragraph{Plasmons:} We follow the discussion of Ref.~\cite{Pines1962Plasmons} to calculate the contribution of the plasmon to the free energy. 
\def\wpl{\omega_{\text{pl}}}
The contribution of the plasmon to $\Omega_{\xr}$ is 
\begin{equation}
    \frac{\Omega_{\text{pl}}}{\A}=  - \int_{\abs{\qbs}<q_c } \frac{\dd^2{\qbs}}{(2\p)^2} \int_{\l_c}^1 \frac{\dd{\l}}{\l} \int_0^{\infty} \frac{\dd{\omega}}{2\p}\coth(\tfrac{\b\w}{2})\Im\left[\frac{1}{\pdv{\kappa'}{\omega}\eval_{\omega = \wpl} \times \left(\omega -\wpl(\qbs;\l) \right)+ \ii  \kappa''}\right],
\end{equation}
where $\l_c$ is a cutoff coupling constant at which the plasmon starts to become damped, $\kappa' = 1 + \l \tilde{v}\Pi'$ and $\kappa'' = \l \tilde{v}\Pi''$ are the real and imaginary part of the dielectric constant. Note that $\wpl$ is the solution to 
\begin{equation}
    \kappa' = 1+ \l\tilde{v}(\qbs)\Pi'(\qbs,\wpl(\l)) =0.
\end{equation}
We then take derivatives with respect to $\l$ to obtain 
\begin{equation}
    \tilde{v}\pdv{\wpl}{\l} \pdv{\Pi' }{\omega}= \frac{1}{\l^2}.
\end{equation}
Furthermore, noting that $\pdv{\k'}{\w} = \l\tilde{v}\pdv{\Pi'}{\omega}$, we can replace the integration over the coupling constant $\l$ by integration over the plasmon frequency to obtain 
\begin{equation}
    \begin{split}
        \frac{\Omega_{\text{pl}}}{\A}
        &= 
     - \int_{\abs{\qbs}<q_c } \frac{\dd^2{\qbs}}{(2\p)^2} \int_{\Omega_c}^{\wpl} {\dd{\Omega}} \int_0^{\infty} \frac{\dd{\omega}}{2\p}\coth(\tfrac{\b\w}{2})\left[-\p \delta(\w - \Omega)\right]\\
     & = 
      \frac{1}{\b}\int_{\abs{\qbs}<q_c } \frac{\dd^2{\qbs}}{(2\p)^2} \left[\log(2\sinh(\tfrac{\b\wpl(\qbs)}{2})) - 
      \log(2\sinh(\tfrac{\b\Omega_{\text{c}}}{2}))\right].
    \end{split}
\end{equation}
Here $\Omega_{\text{c}}$ is a cutoff frequency at which the plasmon starts to be damped. As in Ref.~\cite{Pines1962Plasmons}, we associate the term involving $\Omega_{\text{c}}$ because it is related to electron-hole continuum. The term involving $\wpl$ is the contribution to the thermodynamic potential of a boson with dispersion $\wpl$. Due to the dispersion $\wpl \propto \sqrt{\abs{\qbs}}$, the temperature dependence is going to be of order $T^5$, so it is negligible compared to the electron-hole contribution. 

\paragraph{T-dependence of Lindhard function:}

Finally, the contribution $\varpi^{\mathfrak{R}}_{\xr}$ comes from the temperature dependence of the Lindhard function. To estimate this contribution, we use a Sommerfeld expansion for $\Pi(\qbs,\ii \eps; \m ,T) = \Pi(\ii \eps,\qbs; \m ,T) + \frac{\p^2}{6(\b\m)^2} {\Delta\Pi}(\ii \eps,\qbs; \m) + \Omc(T^4)$ which is valid for $\eps \neq 0$. Then
\begin{equation}
\begin{split}
    \varpi^{\mathfrak{R}}_{\xr} &\approx - \frac{m}{2\p }\frac{\p^2}{6 \b^2}  \Ymf;\\
    \Ymf &= \frac{2\p}{m \mu^2 }\int \frac{\dd^2{\qbs}}{(2\p)^2} \int \frac{\dd{\eps}}{2\p} \frac{\tilde{v}(\qbs){\Delta\Pi}(\ii\eps,\qbs;\m)}{1+\tilde{v}(\qbs) \Pi(\ii\eps,\qbs;\m)}.
\end{split}
\end{equation}
Note that $\Ymf = \Ymf_a(\rs,\h)$ depends on $\rs$, $\h$ and the state $a\in \StateList$.

For a single isotropic flavor, we utilize the change of variables $\abs{\qbs} = 2\kF x$ and $\epsilon = \frac{2\kF^2}{m}xy$ to obtain
\begin{equation}\label{eq:app:Ymf0}
    \Ymf_0(\rs) \equiv \frac{16}{\p}\int_0^{\infty}\dd{x}\int_0^{\infty}\dd{y} x^2  \frac{\Delta\pi(x,y)}{\tfrac{4}{\rs}x+ \pi(x,y)};
\end{equation}
where $\pi(x,y) = 1 - \Re[\sqrt{(x+\ii y)^2-1}]/x$ and $\Delta \pi(x,y) =   \frac{1}{4x} \Re[((x+\ii y)^2-1)^{-3/2}]$. By direct numerical evaluation, we found that  $\Ymf \approx -\rs/\pi$ to leading order in $\rs \ll 1$. 

For the SVP state, we obtain
\begin{equation}
    \Ymf_{\text{SVP}}(\rs,\h) = \int_0^{2\p} \frac{\dd{\q}}{2\p} \frac{\Ymf_{0}(\rs I(\h,\q))}{I(\h,\q)^2}
\end{equation}
We also have $\Ymf_{\text{VP}}(\rs,\h) = \Ymf_{\text{SVP}}(\sqrt{8}\rs,\h)$ and $\Ymf_{\text{Sym}}(\rs,\h) = \Ymf_{\text{SP}}(\sqrt{8}\rs,\h)$.

\paragraph{Final result:}

To summarize our result, the ring exchange contribution to leading order in $T$ for the four states ($a\in \{\text{Sym},\text{SP},\text{VP},\text{SVP}\}$) can be summarized as
\begin{equation}
    \frac{\Omega_{\xr;a}[\m,\h]-E_{\xr;a}[\nel,\h]}{\A} =  -\frac{m}{2\p} \frac{T^2 \p^2}{6}\gmf_{\xr}^{a}(\rs,\h)
\end{equation}
where  $\gmf_{\xr}^{a}\equiv\Xmf_a + \Ymf_a$ is the sum of the real and imaginary contributions. For any of the four states, we found that $\gmf_\xr$ is of order $\rs \log(1/\rs)$ for small $\rs$ but becomes negative for $\rs$ of order 1. 

\begin{figure}[h]
    \centering
    \includegraphics{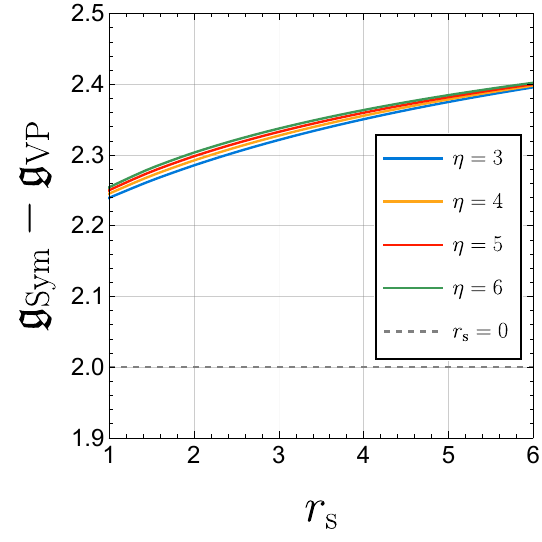}
    \caption{Plot of the coefficient in front of $ \frac{\pi^2 T^2}{6}$ of the difference between the Helmholtz free energy of the symmetric state and valley-polarized state, as a function of $\rs$ at certain values of $\eta$. }
    \label{fig:degeneracyChange}
\end{figure}

\section{Inequalities for ring diagrams}\label{app:IneqRingDiagrams}

\subsection{Leading order ring diagram}\label{app:Er_Leading}

The contribution to the energy of the leading order ring diagram (diagram \textbf{b} of Fig.~\ref{fig:2PT}) is 
\begin{equation}\label{eq:ringDiagram:App}
    \frac{E_{\textbf{b}}}{\A}= -\frac{1}{2}\int_{\RR^2}\frac{\dd^2{\qbs}}{(2\pi)} \int_0^{\infty}\frac{\dd{\Omega}}{2\pi} [\tilde{v}(\qbs)]^2 [\Pi(\ii\Omega,\qbs)]^2.
\end{equation}

We want to show that $\frac{E_{\textbf{b}}}{\A}$ can be thought of as the magnitude of a vector. To do so, we define a variable $\bar{q} \equiv (\qbs,\Omega)$ that lives on the upper-half hyperplane $(\HH^3\equiv \RR_+\times \RR^2)$. Then, we define an inner-product between two real functions $f_{1/2}:\HH^3\to \RR $ as 
\begin{equation}
    \expval{f_1,f_2} \equiv \int_{\HH^3}f_1(u)f_2(u)\dd{\upmu_u}
\end{equation}
where $\dd{\upmu_u} = \frac{\dd^2\qbs}{(2\pi)}\frac{\dd{\Omega}}{2\pi}$ is the integration measure. This is clearly a generalization of the standard dot product of $\RR^n$ where the indices now correspond to points in $\HH^3$.

We define $\Qmc(u) \equiv  \tilde{v}(\qbs) \Pi(\qbs,\ii\Omega)$, where $u = (\qbs,\Omega)$. We can then rewrite Eq.~\ref{eq:ringDiagram:App} as
\begin{equation}
    \frac{E_{\textbf{b}}}{\A} = -\frac{1}{2}\expval{\Qmc,\Qmc}.
\end{equation}
Consider now the partially polarized spin and partially polarized valley states at the same polarization $\zeta_v=\zeta_s = \zeta$. The Lindhard function $\Pi(\ii\Omega,\qbs)$ for the partially valley and spin polarized states (as defined in Tab.~\ref{tab:DefinitionStates}) can be written as 
\begin{equation}
    \begin{split}
        \Pi_{\text{pVP}}(\ii\Omega,\qbs;\zeta) &= 2\pi_0(\qbs,\Omega;\zeta)\\
        \Pi_{\text{pSP}}(\ii\Omega,\qbs; \zeta) &= \pi_0(\qbs,\Omega;\zeta) + \pi_0(C_4\qbs,\Omega;\zeta)
    \end{split}
\end{equation}
where 
\begin{equation}
    \begin{split}
        \pi_0(\qbs,\Omega;\zeta) 
        &= \frac{m}{2\pi} \left(1- \Re[\sqrt{(1+ \ii\Omega /\eps_1(\qbs) )^2 - (2\kF^2/m\eps_{1}(\qbs)) \frac{(1+\zeta)}{2}}]\right) \\
        &\quad + \frac{m}{2\pi} \left(1- \Re[\sqrt{(1+ \ii\Omega /\eps_3(\qbs) )^2 - (2\kF^2/m\eps_{3}(\qbs)) \frac{(1-\zeta)}{2}}]\right)
    \end{split}
\end{equation}
is the polarization bubble from a Slater determinant of the $\a=1,3$ states at densities $n_{\a=1} =\frac{\kF^2}{4\p}\frac{(1+\zeta)}{2}$ and $n_{\a=3} =\frac{\kF^2}{4\p}\frac{(1-\zeta)}{2}$. Recall that $\eps_1(\qbs) = \frac{q_x^2/\sqrt{\h}+ q_y^2\sqrt{\h}}{2m}$ and $\eps_3(\qbs) = \eps_3(C_4\qbs)$. 

Let $\Qmc_0(\qbs,\Omega;\zeta)= \tilde{v}(\qbs)\pi_0(\qbs,\Omega; \zeta)$, then 
\begin{equation}
    0< \expval{\Qmc_0 -C_4\Qmc_0,\Qmc_0 -C_4\Qmc_0} = 2\left(\expval{\Qmc_0,\Qmc_0}-\expval{\Qmc_0,C_4\Qmc_0}\right)
\end{equation}
as long as $\Qmc_0\neq C_4\Qmc_0$, where$(C_4\Qmc_0)(\qbs,\Omega;\zeta)\equiv \Qmc_0(C_4\qbs,\Omega;\zeta)$, which is true for generic $\eta\neq 1$ and $\tilde{v}$ not a constant. 

On the other hand 
\begin{equation}
    \begin{split}
        \frac{E_{\textbf{b}}^{\text{pSP}}}{\A} &= -\frac{1}{2} \expval{\Qmc_0 + C_4\Qmc_0,\Qmc_0 + C_4\Qmc_0}\\
        &= -\left(\expval{\Qmc_0,\Qmc_0}+ \expval{\Qmc_0,C_4\Qmc_0}\right)\\
        \frac{E_{\textbf{b}}^{\text{pVP}}}{\A} &= -2\expval{\Qmc_0 ,\Qmc_0}\\
        \Rightarrow\frac{E_{\textbf{b}}^{\text{pSP}}-E_{\textbf{b}}^{\text{pVP}}}{\A} &= \expval{\Qmc_0,\Qmc_0}-\expval{\Qmc_0,C_4\Qmc_0} >0
    \end{split}
\end{equation}

\subsection{Random Phase Approximation (ring diagrams)}\label{app:Er_RPA}

We rephrase Theorem \ref{theorem:Valley} here:
\begin{theorem}\label{theorem:Valley:app}
    The RPA energy correlation energy (defined in Eq.~\ref{eq:RPAenergy}) for the pVP and pSP states at the same polarizion $\zeta_v=\zeta_s=\zeta$ satisfy
    \begin{equation}
        \frac{E^{\mathrm{pVP}}_{\rr}}{\Area} < \frac{E^{\mathrm{pSP}}_{\rr}}{\Area}
    \end{equation}
    as long as $\tilde{v}(\qbs) = \tilde{v}(C_4\qbs)$, $\tilde{v}(\qbs)>0$, $\eta \neq 1$.
\end{theorem}
\begin{proof}
Using the same notation as in App.~\ref{app:Er_Leading}, we have 
\begin{equation}
\begin{split}
    \frac{E^{\mathrm{pSP}}_{\rr}}{\A}&= -\int_{\HH^3} f( \Qmc_0+C_4\Qmc_0)\dd{\upmu}\\
    \frac{E^{\mathrm{pVP}}_{\rr}}{\A}&= -\int_{\HH^3} f( \Qmc_0+\Qmc_0)\dd{\upmu}\\
\end{split}
\end{equation}
where $f(x) = x-\log(1+x)$ is a convex function on the interval $[0,+\infty]$. Notice that $\Qmc_0$ and $C_4\Qmc_0$ only take values on $[0,+\infty]$ and are implicitly functions of $\zeta$. We can thus use Lemma~\ref{Lemma:inequality} with $Z_1 = \Qmc_0$ and $Z_2=C_4\Qmc_0$ to obtain
\begin{equation}
\int_{\HH^3}f(\Qmc_0+C_4\Qmc_0)\dd{\upmu} \leq \int_{\HH^3}f(2\Qmc_0)\dd{\upmu} \Longrightarrow \frac{E^{\mathrm{pSP}}_{\rr}}{\A} \geq \frac{E^{\mathrm{pVP}}_{\rr}}{\A}.
\end{equation}

We move to the strict inequality. For this, notice that $Z_1, Z_2$ and $f$ are continuous functions. It is clear from the different dispersion of the $X$ and $Y$ valleys that (as long as $\eta \neq 1$), there is a compact region of finite volume in $\HH^3$ such that $Z_1 \neq Z_2$. Let this region be $R$. Let $\delta $ be the smallest value of $j =F \circ(Z_1,Z_2)$ in $R$, then $\upmu(J_{\delta}) \geq \upmu(R)>0$. Then using the second part of the lemma, we get the strict inequality. 
\end{proof}

\subsection{Useful Lemma}\label{app:Lemma}

We state a useful lemma in terms of measure theory. For people not familiar with this, one can restrict $M$ to be $\RR_{+}\times \RR^2$ (frequency and momenta) and interpret the integrals as usual Riemann integration with the replacement $\dd{\mu} = \dd{q_0}\dd{q_1}\dd{q_2}$. We decided to present the lemma in this way because it will applicable to finite temperature where one should take $M= \ZZ\times \RR^2$ where the first factor is a discrete set corresponding to Matsubara frequencies. 
\begin{lemma}\label{Lemma:inequality}
    Fix a topological measure space $(M,\dd\upmu)$ and a pair of integrable functions $Z_1,Z_2: M \to [0,\infty)$. Assume further that the measure is invariant under some map $g:M\to M$ and that $Z_1\circ g = Z_2$. Let $f: [0,\infty) \to \RR$  be convex function. Then 
    \begin{equation}\label{eq:lemma0}
        \int_{M} f(Z_1+Z_2)\dd\upmu \leq  \int_{M} f(2Z_1)\dd\upmu
    \end{equation}
    assuming that both integrals exist. Furthermore, if $f\geq 0$ and $\upmu(J_{\delta})>0$ for some $\delta>0$ with $J_{\delta}=\{m\in M | f(2Z_1(m))+f(2Z_2(m))-2f(Z_1(m)+Z_2(m)) \geq \delta\}$, then \ref{eq:lemma0} is a strict inequality. 
\end{lemma}
\begin{proof}
    Let $\RR_+ = [0,\infty)$ and define $F: \RR_+\times\RR_+ \to \RR$ by the formula $F(x,y) = \frac{f(2x)+f(2y) - 2f(x+y)}{2}$. By the convexity of $f$, we have $F\geq 0$. Then if we define $j: M \to \RR$ as $j= F\circ (Z_1,Z_2)$, we get $\int_M j \dd{\upmu} \geq 0$. Expanding out the definition of $j$ we get 
    \begin{equation}\label{eq:Lemma1}
        \int_M \left(f(2Z_1)+f(2Z_2)-2f(Z_1+Z_2)\right) \dd{\upmu} \geq 0. 
    \end{equation}
     Next, using the invariance of the measure under $g$, we have $\int_{M}f(2Z_2) \dd{\upmu}=\int_{M}f\circ (2 Z_1) \circ g\dd{\upmu}=\int_{M}f(2Z_1) \dd{\upmu}$. 

     Therefore, we can write Eq.~\ref{eq:Lemma1} as 
     \begin{equation}\label{eq:Lemma2}
        \int_M \left(f(2Z_1)-f(Z_1+Z_2)\right) \dd{\upmu} =\frac{1}{2}\int_M j \dd{\upmu} \geq  0 .
    \end{equation}

    Using the second condition, we can write $\int_M j \dd{\upmu} \geq \int_{J_{\delta}} \delta \dd{\upmu}  = \delta \,\,\upmu(J_{\delta})>0$ which proves the strict inequality version of Eq.~\ref{eq:lemma0}.
\end{proof}

\section{Useful integrals}
This appendix contains various useful integrals used in other appendices. The following two formulas are well-known
\begin{equation}\label{eq:Integral1}
    \int_{0}^{\infty} \frac{\sin(ax)\sin(bx)}{x}\dd{x} = \frac{1}{2} \log\abs{\frac{a+b}{a-b}}
\end{equation}
\begin{equation}\label{eq:Integral2}
        \int_{0}^{\infty} \frac{\dd{\s}}{\s}J_1(a \s)J_1(b \s) = \frac{1}{2} \frac{\min(a,b)}{\max(a,b)}; a,b>0
\end{equation}

For the next formula, starting in Equation \textbf{2.261} of Ref.~\cite{gradshteyn2014table},replace $R= \abss^2 +2\abss\cdot \bbs x+ \bbs^2 x^2$: 
\begin{equation}\label{eq:Integral3}
    \int \frac{\dd{\xi}}{\abs{\abss + \bbs \xi}} = 
    \frac{1}{\abs{\bbs}} \arcsinh\left(\frac{\bbs\cdot(\abss + \xi \bbs)}{\abs{\bbs}\abs{\abss\times\bbs}}\right) = \frac{1}{\abs{\bbs}} \log\left(\frac{\abs{\bbs}\abs{\abss+\xi \bbs} +\bbs\cdot(\abss+\xi \bbs)}{\abs{\abss\times\bbs}}\right); [\abss,\bbs\in \RR^3, \abss\times \bbs \neq 0]
\end{equation}

The following indefinite integral was found by massaging Mathematica's output:
\begin{equation}\label{eq:Integral4}
    \int \frac{\dd{\xi}}{1+\xi^2} \frac{1}{\sqrt{a+2b\xi + c \xi^2}} 
    = \frac{1}{2} \left( \frac{\arctanh \left(\frac{\sqrt{-a+\ii 2b +c}\sqrt{a+2b \xi +c \xi^2}}{\ii a+ b +\xi(\ii b+ c)}\right)}{\sqrt{-a+\ii 2b +c}} +( \ii \to -\ii) \right)
\end{equation}

Let $\abss,\bbs\in \RR^n$. Setting $c =1 + \bbs^2$, $b= \abss\cdot \bbs$ and $a= \abss^2$ in Eq.~\ref{eq:Integral4}, we arrive at following indefinite integral
\begin{equation}\label{eq:Integral5}
    \int \frac{\dd{\xi}}{1+\xi^2} \frac{1}{\sqrt{\xi^2+(\abss + \bbs \xi)^2}} 
      = \frac{1}{2} \left( \frac{\arctanh \left(\frac{\sqrt{1+(\ii\abss  +\bbs)^2}\sqrt{\xi^2+(\abss + \bbs \xi)^2}}{\abss\cdot (\ii\abss+ \bbs) + \xi (1+ \bbs\cdot (\bbs + \ii\abss))}\right)}{\sqrt{1+(\ii\abss  +\bbs)^2}} +( \ii \to -\ii) \right).
\end{equation}
Using Eq.~\ref{eq:Integral5}, we can obtain the following definite integral 
\begin{equation}\label{eq:Integral6}
    \int^{+\infty}_{-\infty} \frac{\dd{\xi}}{1+\xi^2} \frac{1}{\sqrt{\xi^2+(\abss + \bbs \xi)^2}} 
      = \left( \frac{\arctanh \left(\frac{\sqrt{1+(\ii\abss  +\bbs)^2}\sqrt{1+\bbs^2}}{ (1+ \bbs\cdot (\bbs + \ii\abss))}\right)}{\sqrt{1+(\ii\abss  +\bbs)^2}} +( \ii \to -\ii) \right); \abss,\bbs\in \RR^n.
\end{equation}

\section{Evaluation of susceptibility for T-matrix}\label{app:Tmatrix}

We present two different calculations of the susceptibility appearing in the T-matrix. 
\def\knot{k_{\rm{o}}}
\subsection{Hard UV cutoff}
We evaluate the particle-particle susceptibility referred in the main text. The definition is given by
\begin{equation}
    \Gamma_{\a \a'}= \int _{R_{\a\a'}}\frac{\dd^2\kbs}{(2\pi)^2} \int_{-\infty}^{+\infty} \frac{\dd\omega}{2\pi} \frac{1}{+\ii \omega-\e_{\a,+\kbs}} \frac{1}{-\ii \omega-\e_{\a',-\kbs}}
\end{equation}
where $R_{\a\a'}$ is a region of $\RR^2$ such that $\eF\lesssim \e_{\a,\kbs},\e_{\a',\kbs}\lesssim W$, and $\varepsilon_{\a,\kbs}\equiv \frac{\kbs\cdot W_\a \cdot \kbs}{2m}$ where $W_\a$ is the anisotropy matrix ($W_\a = \diag(1/\sqrt{\h_\a},\sqrt{\h_\a})$). For simplicity, we approximate $R_{\a\a'}$ by the region $2\eF<\e_{\alpha,\kbs} +\e_{\alpha',\kbs}<2W$.

Performing the integral over $\omega$ gives 
\begin{equation}
    \Gamma_{\a \a'}= \int _{R_{\a\a'}}\frac{\dd^2\kbs}{(2\pi)^2}\frac{1}{\e_{\a,+\kbs}+\e_{\a',-\kbs}}.
\end{equation}
We now distinguish two cases. First, when $W_{\a} = W_{\a'}$, $\e_{\a',-\kbs}= \e_{\a,\kbs}$, the integral is 
\begin{equation}\label{eq:Intravalley}
    \Gamma_{\a \a'}= \frac{m}{4\pi}\log(W/\eF); \quad [\t_\a = \t_{\a'}].
\end{equation}
For the case where $\t_\a =X$ and $\t_{\a'}=Y$, the denominator becomes $\frac{\kbs^2}{2m}(\sqrt{\eta}+1/\sqrt{\eta}) $. 
\begin{equation}\label{eq:Intervalley}
    \Gamma_{\a \a'}= \frac{m}{4\pi}\log(W/\eF)\times \frac{2}{\eta^{+1/2}+\eta^{-1/2}}; \quad [\t_{\a}\neq \t_{\a'}].
\end{equation}

We can write Eqs.~\ref{eq:Intravalley} and \ref{eq:Intervalley} compactly as 
\begin{equation}
    \Gamma_{\a\a'} = \frac{1}{2\pi} \frac{1}{\sqrt{\det(W_{\a}+W_{\a'})}} \log(W/\eF).
\end{equation}
This tell us that the correct mass in the prefactor of $\Gamma$ is the reduced mass.

\subsection{Soft UV cutoff }

An alternative way to calculate $\Gamma_{\a\a'}$ is to impose a soft IR cutoff by modifying the large momentum dependence of the dispersion:
\begin{equation}
    \begin{split}
        \varepsilon_{X,\kbs} &= \frac{k^2}{2m}\left( \cosh(\s) - \sinh(\s)\cos(2\q) +  \frac{k^2}{\Lambda^2}\right)\\
        \varepsilon_{Y,\kbs} &= \frac{k^2}{2m}\left(  \cosh(\s) + \sinh(\s)\cos(2\q)+  \frac{k^2}{\Lambda^2}\right)
    \end{split}
\end{equation}
where $\h=\ee^{2\s}$ and $\L$ is the UV cutoff. 

We calculate the susceptibility by imposing an (isotropic) hard IR cutoff $\knot$. The intra-valley susceptibility is
\begin{equation}
    \begin{split}
        \Gamma_{XX} 
        &= \int_{\knot}^{\infty} \frac{k\dd{k}}{2\pi} \int_0^{2\pi}\frac{\dd{\q}}{2\pi}\frac{1}{\frac{2k^2}{2m}\left[ \cosh(\s) + \sinh(\s) \cos(2\q)  + k^2/\L^2\right]}\\
        &= \frac{m}{2\pi} \int_{\knot}^{\infty}\frac{\dd{k}}{k} \int_0^{2\p} \frac{\dd{\q}}{2\pi} \frac{1}{\cosh(\s) + \sinh(\s) \cos(2\q) + k^2/\L^2 }\\
        &= \frac{m}{2\pi} \int_{\knot}^{\infty}\frac{\dd{k}}{k} \frac{1}{\sqrt{[\cosh(\s)+k^2/\L^2 ]^2 - \sinh(\s)^2 }}\\
        & \sim \frac{m}{2\pi} \int_{\knot}^{\infty}\frac{\dd{k}}{k} \frac{1}{1+\cosh(\s)k^2/\L^2} \sim \frac{m}{4\pi} \log[\L^2/\knot^2 \cosh(\s)].
    \end{split}
\end{equation}
Similarly, the inter-valley susceptibility is
\begin{equation}
    \begin{split}
        \Gamma_{XY} 
        &= \int_{\knot}^{\infty} \frac{k\dd{k}}{2\pi} \int_0^{2\pi}\frac{\dd{\q}}{2\pi}\frac{1}{\frac{2k^2}{2m}\left[ \cosh(\s)  + k^2/\L^2\right]} \sim \frac{m}{4\pi \cosh(\s)} \log[\cosh(\s)\L^2/\knot^2 ].
    \end{split}
\end{equation}
We thus see that the prefactor of $\log(\L^2/\knot^2)$ for $\Gamma_{XX}$ and $\Gamma_{XY}$ is the same as in the calculation with a hard UV cutoff. 

\subsubsection{Trigonal warping}\label{app:TrigonalWarping}

The soft UV cutoff allows us to include the effect of trigonal warping. Consider a system with valleys at the $\Kbs$ and $-\Kbs$ points (denoted as $+$ and $-$, respectively). The dispersion relations we consider are 
\begin{equation}
    \varepsilon_{\pbs,\pm} = \frac{p^2}{2m} \left[1 \pm t p \cos(3\q) + \frac{p^2}{\L^2}\right],
\end{equation}
where $t$ is the trigonal warping strength and $\L$ is the UV cutoff. 

The intervalley and intravalley susceptibilies are
\begin{equation}
    \begin{split}
        \Gamma_{++} &= \int_{\knot}^{\infty} \frac{k\dd{k}}{2\pi} \frac{\dd{\q}}{2\p} \frac{1}{\varepsilon_{\kbs,+} + \varepsilon_{-\kbs,+}}\\
        \Gamma_{+-} &= \int_{\knot}^{\infty} \frac{k\dd{k}}{2\pi} \frac{\dd{\q}}{2\p} \frac{1}{\varepsilon_{\kbs,+} + \varepsilon_{-\kbs,-}}
    \end{split}
\end{equation}

We first evaluate the intervalley susceptibility:
\begin{equation}
    \Gamma_{++} = m \int_{\knot}^{\infty} \frac{k\dd{k}}{2\pi} \frac{1}{k^2(1+k^2/\L^2)} = \frac{m}{4\pi} \log(1+ [\L/ \knot]^2) \sim \frac{m}{4\pi} \log[ W/\eF]
\end{equation}
where the band width $W \sim \L^2/2m$ and $\eF \sim \knot^2/2m$ and we are in the limit $\L/\knot \gg 1$.

Similarly, the intravalley susceptibility is
\begin{equation}
    \Gamma_{+-} = m \int_{\knot}^{\infty} \frac{k\dd{k}}{2\pi} \frac{1}{k^2\sqrt{(1+k^2/\L^2)^2- (t k )^2}}.
\end{equation}
We assume that $\knot, \Lambda$ and $t$ are such that $\Gamma_{+-}$ is a real number, which is true when $t \knot$ is small. As we are interested in the small momentum behaviour, we expand the square root for small $k$:
\begin{equation}
    \Gamma_{+-} = \frac{m}{2\pi} \int_{\knot}^{\infty}\frac{\dd{k}}{k} \frac{1}{(1+k^2(1/\L^2-t^2/2) ) } \sim \frac{m}{4\pi}\log([\tilde{\L}/\knot]^2)
\end{equation}
with $\tilde{\L}^2 = \frac{\L^2}{1- \L^2t^2/2}$. 

Therefore
\begin{equation}
    \Gamma_{++} - \Gamma_{+-} \sim \frac{m}{2\pi} \log(\frac{\L}{\tilde{\L}}) = \frac{m}{4\pi}\log( 1 - (\L t)^2/2) \sim \mathcal{O}([\log(\L/\knot)]^0).
\end{equation}
Even though they are equal to the leading order in the small parameter, their difference is negative
\begin{equation}
    \Gamma_{++} - \Gamma_{+-} 
    =
    \frac{m}{2\pi}\int_{\knot}^{\infty} \frac{\dd{k}}{k (1+ k^2/\L^2)} \left[1- \frac{1}{\sqrt{1- (tk/(1+k^2/\L^2))^2}}\right]<0.
\end{equation}

\end{document}